\documentclass[11pt,letterpaper]{article}

\usepackage[utf8]{inputenc}
\usepackage[T1]{fontenc}
\usepackage{microtype}

\usepackage[margin=1in]{geometry}

\usepackage[small]{caption}
\usepackage{subcaption}

\usepackage{nicefrac}
\usepackage{amsmath}
\usepackage{amsfonts}
\usepackage{amssymb}
\usepackage{amsthm}
\usepackage{amstext}
\usepackage{thmtools}

\usepackage{csquotes}
\usepackage{xcolor}
\usepackage{graphicx}

\usepackage{mathtools}
\usepackage{xspace}

\usepackage[ruled]{algorithm2e}
\usepackage{paralist}
\usepackage{bm}
\usepackage{dsfont}
\usepackage{thm-restate}
\usepackage{enumitem}
\usepackage{paralist}
\usepackage{comment}
\usepackage{soul}

\definecolor{job1}{HTML}{56b3e9}
\definecolor{job2}{HTML}{e69f00}
\definecolor{job3}{HTML}{009e74}
\definecolor{job4}{HTML}{cc79a7}
\definecolor{job5}{HTML}{d55e00}
\definecolor{job6}{HTML}{0071b2}

\usepackage{tikz}
\usetikzlibrary{arrows}
\usetikzlibrary{fit}
\usetikzlibrary{calc,shapes.geometric, backgrounds}
\usetikzlibrary{decorations.pathreplacing,calligraphy}
\usetikzlibrary{patterns,snakes}
\tikzset{vertex/.style={circle, draw, fill=black, black, inner sep=0pt, minimum width=7pt},
}

\usepackage{url}
\usepackage[colorlinks]{hyperref}
\hypersetup{citecolor=green!60!black,linkcolor=blue!90!black,urlcolor=blue!90!black,breaklinks}
\usepackage[capitalize,noabbrev,nameinlink]{cleveref}

\usepackage[skip=5pt plus 3pt minus 1pt,parfill=0pt]{parskip}
\makeatletter
\def\thm@space@setup{%
  \thm@preskip=\parskip \thm@postskip=0pt
}
\makeatother

\usepackage[
 backend=biber,
 style=alphabetic,
 citetracker,
 hyperref=auto,
 maxcitenames=5, 
 sortcites,      
 sorting=nyt,
 maxbibnames=12, 
 date=year,      
 isbn=false,     
 url=false,      
 doi=false,      
 eprint=false,   
]{biblatex}
\newbibmacro{string+doiurlisbn}[1]{%
  \iffieldundef{doi}{%
    \iffieldundef{url}{%
      #1
    }{%
      \href{\thefield{url}}{#1}%
    }%
  }{%
    \href{http://dx.doi.org/\thefield{doi}}{#1}%
  }%
}

\DeclareFieldFormat%
[article,inbook,incollection,inproceedings,patent,thesis,unpublished]
  {title}{\usebibmacro{string+doiurlisbn}{#1}}

\newtheorem{theorem}{Theorem}

\newtheorem{lemma}[theorem]{Lemma}
\newtheorem{claim}[theorem]{Claim}
\newtheorem{fact}[theorem]{Fact}
\newtheorem{observation}[theorem]{Observation}
\newtheorem{proposition}[theorem]{Proposition}
\theoremstyle{definition}
\newtheorem{definition}{Definition}

\DeclarePairedDelimiter\floor{\lfloor}{\rfloor}

\newcommand{\cO}{\mathcal{O}}

\newcommand{\EX}{\mathbb{E}}
\newcommand{\VAR}{\mathbb{V}}
\newcommand{\ind}{\ensuremath{\mathds{1}}}

\newcommand{\hp}{\bm\hat{p}}

\newcommand{\cE}{\mathcal{E}}

\newcommand{\pr}{\mathrm{Pr}}
\newcommand{\ops}{m}    
\newcommand{\eps}{\varepsilon}

\newcommand{\opt}{\mathrm{OPT}}
\newcommand{\alg}{\mathrm{ALG}}
\newcommand{\vol}{\mathrm{vol}}
\newcommand{\OPT}{\mathrm{OPT}}

\newcommand{\chunklp}{\text{Chunk-LP}}
\newcommand{\chunkdp}{\text{Chunk-DP}}

\newcommand{\ofts}{operation flow time scheduling problem}

\newcommand{\reduced}{\mathcal{I}_{\mathrm{red}}}
\newcommand{\ts}{\tau}
\newcommand{\afull}{J^{\mathrm{full}}}
\newcommand{\apart}{J^{\mathrm{part}}}
\newcommand{\bapart}{\bar{J}^{\mathrm{part}}}

\newcommand{\bA}{\bar{J}}
\newcommand{\bAfull}{\bar{J}^{\mathrm{full}}}

\newcommand{\frontjob}{\mathsf{front}}
\newcommand{\topjob}{\mathsf{top}}
\newcommand{\curclass}{\mathsf{class}}

\newcommand{\unit}{p}

\newcommand\job{} 
\def\job[#1](#2)(#3){%
  \fill[#1,draw=black,line width=0.6] (#2) rectangle ++(#3);
}

\newcommand\jobT{} 
\def\jobT[#1](#2,#3)(#4,#5)(#6){%
  \fill[#1,draw=black,line width=0.6] (#2,#3) rectangle ++(#4,#5);
  \node[text=black] at (#2+#4/2.0,#3+#5/2.0) {#6};
}

\usepackage[colorinlistoftodos,prependcaption,textsize=scriptsize]{todonotes}

\title{\bfseries Online Flow Time Minimization with \\ Gradually Revealed Jobs}

\author{Alexander Lindermayr\thanks{Technische Universität Berlin, Germany. Part of the work was done while the author was a research fellow at the Simons Institute for the Theory of Computing for the program on ``Algorithmic Foundations for Emerging Computing Technologies''.}
    \and 
    Guido Sch\"afer\thanks{Centrum Wiskunde \& Informatica (CWI) and University of Amsterdam, The Netherlands.}
    \and
    Jens Schl\"oter\thanks{Centrum Wiskunde \& Informatica (CWI), The Netherlands. Supported by the Netherlands Organisation for Scientific Research (NWO) through project OCENW.GROOT.2019.015 ``Optimization for and with Machine Learning (OPTIMAL)''.}
    \and
    Leen Stougie\thanks{Centrum Wiskunde \& Informatica (CWI) and Vrije Universiteit Amsterdam, The Netherlands. Supported in part by the Netherlands Organisation for Scientific Research (NWO) through project OCENW.GROOT.2019.015 ``Optimization for and with Machine Learning (OPTIMAL)'' and Gravitation-project NETWORKS-024.002.003.}
}

\date{}

\begin{document}

\maketitle

\begin{abstract}
We consider the problem of online preemptive scheduling on a 
single machine to minimize the total flow time. In clairvoyant scheduling, 
where job processing times are revealed upon arrival, 
the Shortest Remaining Processing Time (SRPT) algorithm is optimal. 
In practice, however, exact processing times are often unknown. 
At the opposite extreme, non-clairvoyant scheduling, in which processing times are 
revealed only upon completion, suffers from strong lower bounds on the competitive ratio. 
This motivates the study of intermediate information models.
We introduce a new model in which processing times are revealed gradually during execution. Each job consists of a sequence of operations, and the processing time of an operation becomes known only after the preceding one completes. This models many scheduling scenarios that arise in computing systems.

Our main result is a deterministic $O(m^2)$-competitive algorithm, where $m$ is the maximum number of operations per job. 
More specifically, we prove a refined competitive ratio in $O(m_1 \cdot m_2)$, where $m_1$ and $m_2$ are instance-dependent parameters describing the operation size structure.
Our algorithm and analysis build on recent advancements in robust flow time minimization (SODA '26), where jobs arrive with estimated sizes. 
However, in our setting we have no bounded estimate on a job's processing time. Thus, we design a highly adaptive algorithm
that gradually explores a job's operations while working on them, and groups them into virtual chunks whose size can be well-estimated. This is 
a crucial ingredient of our result and requires a much more careful analysis compared to the robust setting.
We also provide lower bounds showing that our bounds are essentially best possible.
For the special case of scheduling with uniform obligatory tests, we show that SRPT at the operation level is $2$-competitive, which is best possible.
\end{abstract}

\thispagestyle{empty}

\newpage

\setcounter{page}{1}

\section{Introduction}

We study the fundamental problem 
of online preemptive flow time minimization
on a single machine: $n$ jobs arrive online over time and must be scheduled with preemption.
The objective is to minimize the total flow time, that is, the sum over jobs of their time in the system.
This is a central quality-of-service metric with applications, e.g., in networking, cloud computing, and operating systems.
A classical result for this problem is that SRPT (Shortest Remaining Processing Time), which always runs the job with the least remaining work, is optimal in this setting~\cite{Schrage68}.

In practice, however, we cannot 
generally
assume that a scheduler has access to the exact processing times $p_j$ of jobs, and therefore cannot 
implement SRPT directly.
Already in the 1990s, this led 
to
the study of \emph{non-clairvoyant} algorithms, which only learn a job's processing time once it has completed. 
From a theoretical perspective, this model lies at the other extreme of the spectrum 
and strong lower bounds are known:
no non-clairvoyant algorithm can be $o(\log n)$-competitive~\cite{MotwaniPT94}. 

A major research direction in online scheduling has therefore been to study intermediate models that bridge the gap between knowing everything (clairvoyance) and knowing nothing (non-clairvoyance), by providing information about a job before it completes~\cite{BenderMR04,BecchettiLMP04,AzarLT21,YingchareonthawornchaiT17}.
Only recently, two prominent models were introduced to achieve this goal: 
\begin{itemize}
	\item \textbf{Predictions.} Each job $j$ arrives with a predicted processing time $\hp_j$~\cite{AzarLT21}; this is also called \emph{robust} flow time scheduling.
    Recently,
	Gupta et al.~\cite{GuptaKPW26} gave a best-possible $O(\mu)$-competitive algorithm, 
    where $\mu := \mu_1 \cdot \mu_2$ and $\mu_1 = \max_j \hp_j / p_j$ and $\mu_2 = \max_j p_j / \hp_j$ 
    are the maximum overestimation and underestimation factors.
	A natural critique of this model is that the prediction must be available at arrival, 
    when job and scheduler have not yet interacted.
    Therefore,
	it is unclear how such predictions could be obtained in many real-world applications.
	\item \textbf{$\eps$-Clairvoyance.} The downside of predictions is
	mitigated in the $\eps$-clairvoyant model, where a scheduler
	learns about a job's processing time once a $(1-\eps)$-fraction of its processing time is done (hence an $\eps$-fraction remains)~\cite{YingchareonthawornchaiT17}. Gupta et al.~\cite{GuptaKLSY25} presented a best-possible $\lceil 1/\eps \rceil$-competitive algorithm for all $\eps \in (0,1]$.
	While this model does allow the algorithm to learn 
    about jobs through interaction, 
    the requirement of learning the processing time of a job at a single specific point 
    is a strong assumption that is hard to justify in practice.
\end{itemize}
To overcome these drawbacks, we introduce a new model in which job information is revealed gradually through execution and that has strong connections to both theory and practice. 


From a practical perspective, we move beyond the black-box abstraction and examine what jobs 
in fundamental computing applications are: programs represented by 
control-flow graphs~\cite{AhoSU86, Allen70, AllenC76}. In system design and architecture, 
it is well established that programs exhibit distinct phase behavior, 
executing sequences of basic blocks or operations over time~\cite{SherwoodPHC02}.
We model this structure formally using a decision tree for each job. 
We assume that a scheduler can determine the execution time of any deterministic operation \emph{on a given input}~\cite{King76}. 
Initially, it knows the processing time of the root operation (or phase). When the root 
completes, the program branches based on its result and reveals the next operation, along with its processing time.
This repeats until the job 
completes. For deterministic programs, an offline optimum knows 
the entire sequence of operations and processing times in advance. See \Cref{fig:decision-tree-operations} for an illustrative example.

Formally, we consider the following
\emph{\ofts}.
A job $j$ is composed of 
$\ops$ \emph{operations} $j_1,\ldots,j_\ops$ with operation processing times
$p_{j_1},\ldots,p_{j_\ops}$ and job processing time $p_j = \sum_{\ell=1}^\ops p_{j_\ell}$. 
When job $j$ arrives at time $r_j$, we say that $j$ and its first operation $j_1$ become active and known. Once the first operation is completed, the next operation $j_2$ becomes active and known to the scheduler, and so on, until all operations are completed and the job completes.
We give a precise definition later.
It is helpful to think of operations not as single instructions, but as small, linear subprograms 
that need to be completed
before the overall program branches.

\begin{figure}[t]
    \centering
    \begin{tikzpicture}[
        grow=right,
        level 1/.style={sibling distance=1.8cm, level distance=3cm},
        level 2/.style={sibling distance=0.8cm, level distance=3cm},
        level 3/.style={sibling distance=1.2cm, level distance=3cm},
        level 4/.style={sibling distance=0.8cm, level distance=3cm},
        op/.style={draw, rectangle, minimum height=0.6cm, fill=white, font=\small, inner sep=2pt, line width=0.8pt},
        highlight/.style={fill=job1, draw=black, line width=1.2pt},
        path edge/.style={job2, ultra thick, ->, >=stealth},
        normal edge/.style={black, thick, ->, >=stealth},
        uncertain/.style={dashed, gray, thick, ->, >=stealth},
        ghost/.style={draw, rectangle, minimum height=0.6cm, dashed, gray, font=\small, inner sep=2pt, line width=0.8pt},
        scale=0.95
    ]
        \node [op, minimum width=1.5cm, highlight] (root) {$j_1$}
            child {node [ghost, minimum width=1.6cm] {?}
                child {node [ghost, minimum width=1cm] {?}
                    child {node [ghost, minimum width=3cm] {?}
                        edge from parent [normal edge]
                    }
                    edge from parent [normal edge]
                }
                edge from parent [normal edge]
            }
            child {node [op, minimum width=2.5cm, highlight] (c2) {$j_2$}
                child {node [ghost, minimum width=1.4cm] {?}
                    edge from parent [normal edge]
                }
                child {node [op, minimum width=1.8cm, highlight] (c3) {$j_3$}
                    child {node [ghost, minimum width=1.2cm] {?}
                        edge from parent [normal edge]
                    }
                    child {node [op, minimum width=2.2cm, highlight] (c4) {$j_4$}
                        child {node (u3) [ghost, minimum width=1.0cm] {?} edge from parent [uncertain]}
                        child {node (u1) [ghost, minimum width=1.5cm] {?} edge from parent [uncertain]}
                        child {node (u2) [ghost, minimum width=0.8cm] {?} edge from parent [uncertain]}
                        edge from parent [path edge] node[above, font=\small, black] {$x_3$}
                    }
                    edge from parent [path edge] node[above, font=\small, black] {$x_2$}
                }
                edge from parent [path edge] node[above, font=\small, black] {$x_1$}
            }
            child {node [ghost, minimum width=1.2cm] (c1) {?}
                child {node [ghost, minimum width=1.0cm] {?} edge from parent [normal edge]}
                child {node [ghost, minimum width=0.8cm] {?} edge from parent [normal edge]}
                edge from parent [normal edge]
            };

        \draw[path edge] ($(root.west)-(2,0)$) -- (root.west) node[midway, above, font=\small, black] {input $x_0$};

    \end{tikzpicture}
    \caption{Illustration of the decision tree model. Operations are rectangles with widths proportional to processing time. Edges are labeled with intermediate results $x_i$ that determine the length of the next operation. The blue operations are realized operations. Currently, the job is at operation $j_4$. Future operations are unknown until the preceding intermediate result is computed.}
    \label{fig:decision-tree-operations}
\end{figure}
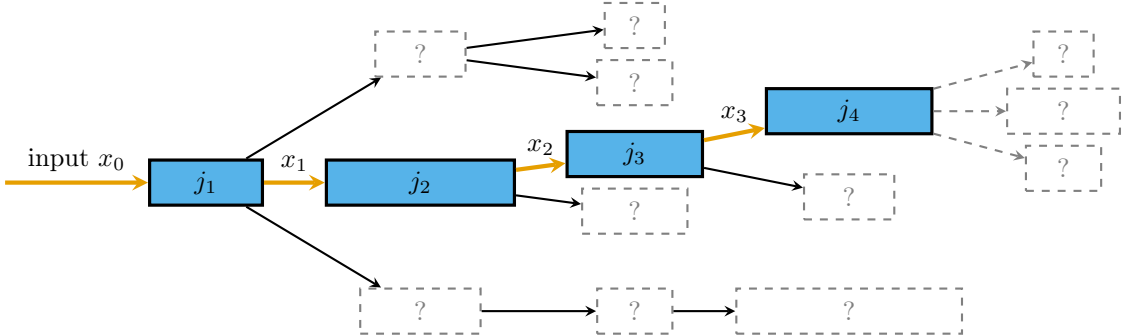

From a theoretical perspective, special cases of 
our \ofts
have deep connections to existing models and results in online flow time scheduling literature:
\begin{itemize}
\item
If all operations have lengths in $\{0,1\}$, a classic result of Motwani, Phillips, and Torng is that running jobs to completion is $m$-competitive and best-possible~\cite{MotwaniPT94}. This is because an optimal solution can finish at most $m$ short jobs while the algorithm works on a long job.

\item If the operation lengths of each job are monotone non-decreasing, 
we show that executing SRPT on the operation level (i.e., at any time schedule the job with the shortest remaining \emph{operation} processing time), which we call \emph{Operations-SRPT}, is $m$-competitive, and this is best-possible (cf.~\Cref{sec:operations-srpt}).
\item If the length of the first operation, which is available at job arrival, underestimates $p_j$ by at most a factor of $m$, we can treat $p_{j_1}$ as a prediction for $p_j$
and obtain a best-possible $O(m)$-competitive algorithm~\cite{AzarLT21,AzarLT22,GuptaKPW26}. 
However, for general instances the length of the first operation can be an arbitrarily bad estimate of the total job size.
\end{itemize}

We capture all of these aspects 
in a single algorithm and analyze it for general instances.

\subsection{Our Results and Techniques}
\label{sec:our:results}

Our first main result is a bounded competitive ratio in terms of the number of operations $\ops$.

\begin{theorem}\label{thm:main1}
	For scheduling with $\ops$ operations, there exists
	a deterministic $O(\ops^2)$-competitive algorithm 
	for minimizing the total flow time on a single machine.
\end{theorem}

Our algorithm is an optimistic, adaptive variant
of the robust algorithm in \cite{GuptaKPW26}.
At a high level, we treat the first operation
length $p_{j_1}$ as a prediction of $p_j$ and run
the robust algorithm using this prediction.
As long as subsequent operations have processing
time at most $p_{j_1}$, the prediction underestimates the total processing time of the revealed part of job $j$ by at most a factor of $\ops$, so we retain the best-possible $\ops$-competitive guarantee.
If an operation $j_i$ with processing time much larger than $p_{j_1}$ becomes active, we stop processing $j$ and effectively
treat it as re-arriving, starting from operation $i$ and with prediction $p_{j_i}$.
We call the 
operations $j_1,\ldots,j_{i-1}$ a \emph{chunk}.

Since the chunks of a job have increasing predictions, handling those has a strong connection to the special case of instances with monotone non-decreasing operation processing times, for which we show that Operations-SRPT is $m$-competitive.

\begin{restatable}{theorem}{thmOperationsSRPT}\label{thm:op-srpt}
	For scheduling with $\ops$ operations, Operations-SRPT is $m$-competitive for minimizing the total flow time on a single machine 
	if $p_{j_1} \leq \ldots \leq p_{j_m}$ for all jobs $j$.
\end{restatable}

These two ingredients suggest the $O(\ops^2)$ bound: we lose a factor $\ops$ by underestimating the size of a chunk by a factor of at most $m$ as the chunk can consist of that many operations, and another factor $\ops$ because each job may be split into up to $\ops$ chunks, treated as $\ops$ virtual jobs.
In fact, by defining 
$\ops_1$ as the maximum number of chunks per job
and
$\ops_2$ as the maximum number of operations per chunk,
we can show a refined bound, which implies \Cref{thm:main1}; 
we give precise definitions in \Cref{sec:general}.
Importantly, $\ops_1$ and $\ops_2$ depend only on the instance,
and not on the algorithm.

\begin{theorem}\label{thm:main2}
	For scheduling with $\ops$ operations, there exists
	a deterministic $O(\ops_1 \cdot \ops_2)$-competitive algorithm 
	for minimizing the total flow time on a single machine.
\end{theorem}

Our analysis builds on \cite{GuptaKPW26} and  uses dual fitting. However, we use a different dual program which is not based on jobs but on chunks. Further, our more adaptive algorithm requires several new ideas and very careful adaptations in the analysis.

In terms of $\ops$, we can only show a (randomized) lower bound of $\Omega(\ops)$, so \Cref{thm:main1} is not tight. However, in terms of 
$\ops_1$ and $\ops_2$, \Cref{thm:main2} is tight in the following sense.
On the one hand, if $\ops_2 = 1$, so each chunk is composed of one operation, the result by \cite{MotwaniPT94} shows that
every deterministic algorithm has a competitive ratio of at least $\ops_1$ (cf.\ \Cref{thm:det-lb}), and we show that every randomized algorithm has a competitive ratio of at least $\Omega(\ops_1)$ (cf.\ \Cref{theorem:randomized-lb}).
On the other hand, if $\ops_1 = 1$, the lower bound for robust flow time implies a competitive ratio of $\Omega(\ops_2)$~\cite{AzarLT22}.

Finally, we look at the case of two operations per job ($m=2$). 
This recovers \emph{scheduling with obligatory tests}~\cite{DogeasEL24}, where a mandatory test determines a job’s processing time (the test is the first operation, and the job is the second). Dogeas et al.~\cite{DogeasEL24} consider total completion time minimization and show that the best possible deterministic competitive ratio is between $\sqrt{2}$ and $1.861$, with an improved upper bound of $1.585$ for uniform-length tests ($p_{j_1}=\unit$).
Our~\Cref{thm:main1} and lower bounds imply that, for total flow time, the best possible deterministic competitive ratio for scheduling with obligatory tests lies between $2$ and $672$; we show it is $2$ for uniform-length tests.

\begin{restatable}{theorem}{thmMainSWT}\label{thm:main-swt}
	For scheduling with obligatory uniform-length tests, Operations-SRPT is 2-competitive for minimizing the total flow time on a single machine.
\end{restatable}

Our analysis of Operations-SRPT for this special case is inspired by Schrage's inductive SRPT analysis~\cite{Schrage68}. While SRPT always processes the smallest job in the system, the same does not hold for Operations-SRPT on the operation level:
If we consider all operations that belong to active jobs (even those operations that are not yet active), then Operations-SRPT does not necessarily process the smallest such operation, because it might not be active yet.
To compensate for this crucial difference to SRPT, we apply a variant of the inductive SRPT analysis only to a subset of operations for which we can guarantee that they are processed in order of their remaining processing times. We then carefully complement this 
with volume-based arguments 
to argue about all operations. 

Our results for scheduling with obligatory tests are in contrast to the model of scheduling with \emph{optional} tests~\cite{DurrEMM20}, for which we show in~\Cref{sec:optional-tests} that no deterministic algorithm has a constant competitive ratio, even for uniform-length tests. The main difference to the model with obligatory tests is that the total processing volume is not equal for all schedules anymore and instead depends on the tests that an algorithm decides to execute, which an adversary can exploit.

\subsection{Organization}

In the next section, we give precise definitions and introduce notation that will
be required throughout this paper. Then, in \Cref{sec:operations-srpt} we introduce
Operations-SRPT and prove \Cref{thm:op-srpt}. In \Cref{sec:swt} we
analyze this algorithm for scheduling with obligatory uniform-length tests.
Finally, in \Cref{sec:general} we present our algorithm for the general 
case and prove \Cref{thm:main1,thm:main2}. Most proofs and all lower bounds are deferred to the appendix. We discuss further related work in~\Cref{app:related}.

\section{Notation and Preliminaries}\label{sec:preliminaries}

We give a formal definition of our \emph{\ofts}. 
There are $n$ jobs, denoted $1,\ldots,n$, that arrive online over time. Each job $j$ arrives at an integer release date $r_j$, has an integer processing time $p_j \geq 1$, and is composed of $m$ operations $j_1,\ldots,j_m$, each with an integer operation processing time $p_{j_i} \geq 0$ such that $p_j = \sum_{i=1}^m p_{j_i}$. 
We call $j_i$ a \emph{stage}-$i$ operation of job $j$.
The more general setting where jobs have \emph{at most} $m$ operations can be easily simulated using zero-length operations.
During each integer interval $[t,t+1]$, an algorithm can process at most one job for one unit of processing.
The \emph{completion time} $C_j$ of job $j$ is the first point in time when it has received a total of $p_j$ processing and its \emph{flow time} is $F_j := C_j - r_j$. The objective is to minimize $\sum_j F_j$.

Let $\alg$ and $\opt$ denote the total flow time of an algorithm and the optimal solution, respectively. An algorithm is \emph{$c$-competitive} if, for any instance, $\alg \leq c \cdot \opt$.
We say that a job is \emph{\textbf{active}} 
at time $t$ if $r_j \leq t < C_j$. 
An algorithm is \emph{locally $c$-competitive} if at any time $t$, it holds that $|J(t)| \leq c \cdot |J^*(t)|$, where $J(t)$ denotes the set of active jobs in the algorithm's schedule at time $t$ and $J^*(t)$ in the optimal solution~\cite{Schrage68}.  
Since $\sum_j F_j = \sum_{t \geq 0} |J(t)|$,
every locally $c$-competitive algorithm is also $c$-competitive.

Let $y_j(t)$ denote the cumulative processing of job $j$ until time $t$.
For each job $j$, we call the operation of $j$ which is next to be processed the \emph{\textbf{active operation}};
formally, operation $j_i$ is active at time $t$ if $j \in J(t)$ and 
$\sum_{i'=1}^{i-1} p_{j_{i'}} \leq y_j(t) < \sum_{i'=1}^{i} p_{j_{i'}}$. Since there is exactly one operation active at time $t$ if and only if $j$ is active at time $t$, we slightly overload notation and use $J(t)$ to denote the set of active operations (if clear from the context).
Moreover we say that an operation $j_i$ is an \emph{\textbf{alive operation}} if it has not been completed until time $t$; formally, operation $j_i$ is alive at time $t$ if $j \in J(t)$ and $y_j(t) < \sum_{i'=1}^{i} p_{j_{i'}}$. Note the difference between active and alive operations; in particular, at most one operation $j_i$ of job $j$ is active at a time $t$, while all operations succeeding $j_i$ (including) are alive. 
At any time $t$, the algorithm knows the processing times of all active operations but the processing times of the remaining alive operations remain unknown. We remark that the algorithm also does not know the value of $m$ upfront.
Finally, we use 
$p_j(t) = p_j - y_j(t)$ for the remaining processing time of a job $j$ at time $t$ and $p_{j_i}(t)$ for the remaining processing time of operation $j_i$ at time $t$. We denote the corresponding quantities in a fixed optimal solution as $p^*_{j}(t)$ and $p^*_{j_i}(t)$.
 
\section{Operations-SRPT}\label{sec:operations-srpt}

We start by studying the arguably most natural algorithm for scheduling with $m$ operations, which 
we call \emph{Operations-SRPT}: we simply run SRPT on the set of active operations.

\begin{quotation}
    \noindent
	\textbf{Operations-SRPT:} At any time $t$, schedule the active operation $q \in J(t)$ with the shortest remaining operation processing time $\min_{q \in J(t)} p_q(t)$. 
\end{quotation}

We break ties in favor of the smaller operation index, i.e., we prefer $j_k$ over $j'_{\ell}$ if $k < \ell$; if $k=\ell$ we prefer $j$ if $j < j'$.
We show that Operations-SRPT is $m$-competitive for
\emph{monotone non-decreasing} operation processing times, which is best-possible (cf. \Cref{thm:det-lb}).

\thmOperationsSRPT*

\begin{proof}
We show that Operations-SRPT is locally $m$-competitive. Fix a time $t$.
We apply that SRPT is locally $1$-competitive~\cite{Schrage68} to a \emph{virtual operations instance} $J_o$: all operations of job $j$ are released and active at time $r_j$ and can be processed in any order.
Let $J_o(t)$ and $J_o^*(t)$ denote the set of active operations at time $t$ in an algorithm's schedule
and in an optimal schedule for $J_o$, respectively.

For instance $J_o$, consider the schedule of SRPT, which treats each operation as an independent job.
Since 
$p_{j_1} \leq \ldots \leq p_{j_m}$ and $j_1,\ldots,j_m$ are available at time $r_j$ for each job $j$,
SRPT would schedule $j$'s operations in order $j_1,\ldots,j_m$. In particular, note that 
the schedule of SRPT on $J_o$ is equivalent to the schedule of Operations-SRPT on $J$ under the same tie breaking.
Using this and that SRPT is locally $1$-competitive on $J_o$, we have
\(
    |J(t)| = |J_o(t)| \leq 1 \cdot |J_o^*(t)|.
\)
Finally, since an optimal solution for instance $J$ is also feasible for instance $J_o$
and for every active job in instance $J$ at time $t$ there can be at most
$m$ active operations (per job) in instance $J_o$ at time $t$, 
we conclude $|J_o^*(t)| \leq m \cdot |J^*(t)|$.
Combining both inequalities implies $|J(t)| \leq m \cdot |J^*(t)|$, which concludes the proof.
\end{proof}

Moreover, we show in \Cref{app:srpt:lb} that monotone non-decreasing operation processing times exactly characterize this algorithm in the following sense.

\begin{restatable}{theorem}{thmSRPTlB}
    \label{thm:srpt:lb}
    The competitive ratio of Operations-SRPT is $\Omega(\log n)$ if $m=2$ and $p_{j_1} > p_{j_2}$ for all $j$.
\end{restatable}

In the next section, however, we will see that for $m=2$ Operations-SRPT is constant-competitive
under the additional assumption that the stage-$1$ operations $j_1$ are all of equal size.

\section{Two Operations: Scheduling with Uniform Obligatory Tests}\label{sec:swt}

In this section, we consider the special case of $m=2$ and uniform-length stage-$1$ operations, that is, $p_{j_1} = \unit$ for a common integer $\unit \in \mathbb{N}_+$. 
This corresponds to 
scheduling with uniform-length obligatory tests~\cite{DogeasEL24}. Our goal is to prove~\Cref{thm:main-swt}, which we restate here for convenience.

\thmMainSWT*

We first classify the jobs based on the size of their stage-$2$ operation.

\begin{definition}[Job types]
    A job $j\in J$ is a \emph{type-A} job if $p_{j_2} \ge \unit$ and a \emph{type-B} job if $p_{j_2} < \unit$.
\end{definition}

This definition ties into the concept of chunks as described in~\Cref{sec:our:results}: The operations of a type-B job $j$ can be considered a single chunk, since $p_{j_1}$ is a $2$-approximation of the total size $p_j$ of the job. For type-A jobs this is not the case, so each operation can be considered an individual chunk.


Next, we state an important property of Operations-SRPT that holds for our special case provided that we break ties in favor of stage-$1$ operations. We defer the proof to~\Cref{app:swt}.

\begin{restatable}{observation}{obsSinglePartial}
  \label{obs:single:partial}
  At any time $t$, at most one job has a remaining processing time of less than $\unit$.
\end{restatable}

Before we move to the competitive analysis, we introduce more notation.
If $j$ is a type-$\gamma$ job, then we also say that $j_1$ and $j_2$ are type-$\gamma$ operations.
For a type $\gamma \in \{A,B\}$ and an operation stage $\ell \in \{1,2\}$, let $Q^*_{\ell \gamma}(t)$ denote the set of alive  \emph{stage-$\ell$ type-$\gamma$ operations} at time $t$ in an optimal solution.
We will also argue about job \emph{volumes}.
For a set of operations or jobs $H$ that are alive at point in time $t$ in the optimal solution, 
let $\vol_t^*(H) = \sum_{q \in H} p^*_q(t)$. Similarly, we define $\vol_t(H) = \sum_{q \in H} p_q(t)$ for a set of operations of jobs $H$.
Moreover, we define $\vol_t = \vol_t(J(t))$ and $\vol^*_t = \vol^*_t(J(t))$.


Our goal is to show Operations-SRPT is locally $2$-competitive, which implies~\Cref{thm:main-swt}.

\begin{lemma}
\label{thm:unit:locally:2}
Let $\ts$ be an arbitrary point in time. Then it holds that $|J(\ts)| \le 2 \cdot |J^*(\ts)|$. 
\end{lemma}
Our proof of 
\Cref{thm:unit:locally:2}
is inspired by the optimality proof for SRPT by Schrage~\cite{Schrage68}.  
At time $\tau$, he considers the remaining volume of the (at most) $|J^*(\ts)|$-largest active jobs 
in the algorithm's schedule, denoted by $\vol_\ts(L(|J^*(\ts)|,t))$, and shows that 
it is at least the total remaining volume in the optimal solution
$\vol_{\ts}^*(J^*(\ts)) = \vol^*_\ts$. 
Since $\vol_\ts = \vol^*_\ts$, this implies $\vol^*_\ts = \vol_\ts \ge \vol_\ts(L(|J^*(\ts)|,\ts)) \ge \vol^*_\ts = \vol_\ts$,
meaning that the $|J^*(\ts)|$-largest active jobs at time $\ts$ in the algorithm's schedule 
contain all of the algorithm's remaining volume $\vol_\ts$.
This can only be if $|J(\ts)| \leq |J^*(\ts)|$.

Inspired by this, we
(i) prove a similar volume invariant as in~\cite{Schrage68} but only for \emph{stage-$2$ type-A operations} and (ii) show that this weaker volume invariant implies~\Cref{thm:unit:locally:2}.
To this end, let $L_{2A}(x,t)$ denote the set of the $x$-largest stage-$2$ type-A operations alive at  time $t$ in Operations-SRPT's schedule and let $\vol_\tau(L_{2A}(x,t))$ denote the total remaining volume of $L_{2A}(x,t)$ at  time $t$. 
We would like to prove the following inequality, which states that the total remaining volume of the $|J^*_{2A}(t)|$-largest stage-$2$ type-A operations alive at $\ts$ in Operations-SRPT's schedule is at least as large as the remaining volume of all stage-$2$ type-A operations alive at time $\ts$ in the optimal solution:
\begin{equation}
\label{eq:volume:eq}
\vol_\tau (L_{2A}(|Q^*_{2A}(\ts)|,\ts)) \ge \vol^*_\ts(Q^*_{2A}(\ts)) \ .
\end{equation}
In the next lemma, we show that this volume invariant indeed implies \Cref{thm:unit:locally:2}.
To illustrate the proof idea, assume for now that the optimal solution does not have alive stage-$1$ type-A operations at time $\ts$.
Using $\vol_\tau = \vol^*_\tau$, Inequality~\eqref{eq:volume:eq} implies $R:= \vol_\tau - \vol_\tau(L_{2A}(|Q^*_{2A}(\ts)|,\ts)) \le \vol^*_{\tau} - \vol^*_\tau(Q_{2A}(\tau)) =: R^*$, i.e., when ignoring the $|Q^*_{2A}(\ts)|$-largest stage-$2$ type-A operations in both schedules, the remaining volume of the algorithm at time $\tau$ is at most as large as the remaining volume in the optimal solution. Since we assume $Q_{1A}^*(\ts)=\emptyset$, the optimal solution needs at least $R^*/(2\unit)$ alive type-B jobs at time $\ts$ to fill the volume $R^*$, as all those jobs have size at most $2\unit$. 
However, the algorithm has at most one job with a remaining volume smaller than $\unit$ (cf.~\Cref{obs:single:partial}). Hence, the algorithm can fit at most  $R/\unit \le R^*/\unit$ jobs into the volume $R$, which implies $|J(\ts)| \le 2 \cdot |J^*(\ts)|$. In~\Cref{app:swt}, we turn this idea into a full proof that takes stage-$1$ type-A operations into account.

\begin{restatable}{lemma}{lemVolumeConversion}
    \label{lem:volume:conversion}
    If $\vol_\tau(L_{2A}(|Q^*_{2A}(\ts)|,\ts)) \ge \vol^*_\tau (Q^*_{2A}(\ts))$, then $|J(\ts)| \le 2 \cdot |J^*(\ts)|$.
\end{restatable}

Given~\Cref{lem:volume:conversion}, it only remains to prove that the volume invariant~\eqref{eq:volume:eq} holds. A common approach (see e.g.~\cite{Schrage68,BansalD07,AzarLT21}) for this 
is to prove that  
\begin{equation}
\label{eq:volume:inductive}
\vol_t(L_{2A}(|Q^*_{2A}(t) \cap Q^*_{2A}(\ts)|,t)) \ge \vol^*_t(Q^*_{2A}(t) \cap Q^*_{2A}(\ts))
\end{equation}
holds at any time $0 \le t \le \ts$ via induction, which then directly implies~\eqref{eq:volume:eq} for $t = \ts$.

A main difference between our volume invariant and previous approaches is that our invariant is not based on jobs but on  \emph{subsets} of operations.
To illustrate the additional challenges caused by this difference, 
consider the job-based invariant 
$\vol_t(L(|J^*(t) \cap J^*(\ts)|,t)) \ge \vol^*_t(J^*(t)\cap J^*(\ts))$, where $\vol_t(L(|J^*(t) \cap J^*(\ts)|,t))$ is the remaining volume of the $|J^*(t) \cap J^*(\ts)|$-largest alive \emph{jobs} at time $t$ in the algorithm's schedule. 
 The following argument is often crucially used when proving such an 
invariant (see e.g.~\cite{Schrage68,BansalD07,AzarLT21}):
If $|J(t)| \le |J^*(t) \cap J^*(\ts)|$, then the remaining volume of the $|J^*(t) \cap J^*(\ts)|$-largest alive jobs at time $t$ is $\vol_t( L(|J^*(t) \cap J^*(\ts)|,t)) = \vol_t$ as the set of the $|J^*(t) \cap J^*(\ts)|$-largest jobs alive at time $t$ contains \emph{all} jobs that are alive at time $t$. 
As $\vol_t = \vol_t^* \ge \vol^*_t(J^*(t)\cap J^*(\ts))$, this then gives $\vol_t(L(|J^*(t) \cap J^*(\ts)|,t)) \ge  \vol^*_t(J^*(t)\cap J^*(\ts))$.

Unfortunately, this type of argument does not apply to the operation-based volume invariant~\eqref{eq:volume:inductive}: 
Even if the number of alive stage-$2$ type-A operations at time $t$ is at most $|Q^*_{2A}(t) \cap Q^*_{2A}(\ts)|$, 
we can have $\vol_t(L_{2A}(|Q^*_{2A}(t) \cap Q^*_{2A}(\ts)|,t)) < \vol_t$ as there might be type-B and stage-$1$ type-A operations in the system.
In particular, $\vol_t(L_{2A}(|Q^*_{2A}(t) \cap Q^*_{2A}(\ts)|,t)) < \vol_t$ can hold at a point in time $t$ where the algorithm processes an operation $q \in L_{2A}(|Q^*_{2A}(t) \cap Q^*_{2A}(\ts)|,t)$, which is a situation where the argumentation above would typically be used. 
To address this additional challenge, we identify a certain time $t_0$ at which the invariant~\eqref{eq:volume:inductive} is guaranteed to hold. 
If the situation sketched above does not happen during $[t_0,\tau]$, 
we show that~\eqref{eq:volume:inductive} holds for any $t \in [t_0,\ts]$ by essentially replicating the inductive proof of Schrage's SRPT analysis~\cite{Schrage68}, 
but starting the induction at time $t_0$ instead of time $0$. 
Otherwise,
we follow a different proof strategy and show that the case is certainly not a worst-case w.r.t.~the local competitive ratio at time $\ts$. We give the full proof in~\Cref{app:swt}.

\section{Algorithm for General Instances}
\label{sec:general}

In this section, we give our results for general instances and prove \Cref{thm:main1,thm:main2}.
If the operation sizes of jobs are non-decreasing, 
then Operations-SRPT is $O(m)$-competitive (cf.~\Cref{thm:op-srpt}).
If $p_{j_1}$ is a good approximation of $p_j$, specifically if $p_{j_1} \geq \frac{1}{m} p_j$, 
then we can treat $p_{j_1}$ as a prediction and are $O(m)$-competitive~\cite{AzarLT21,GuptaKPW26}.
But what if neither of these scenarios applies? 

Our solution carefully combines these two approaches: 
we first find the largest index $i$ such that $\lfloor \log_2 p_{j_\ell} \rfloor \leq \lfloor \log_2 p_{j_1} \rfloor$ for all $\ell \in \{1,\ldots,i\}$. 
We call the resulting set of consecutive operations $\{j_1,\ldots,j_i\}$ a \emph{chunk}. 
We treat this chunk as a virtual job (as in Operations-SRPT) and essentially run the algorithm of \cite{GuptaKPW26} on it.
When the chunk is completed, we compute the next chunk of the job and repeat.
We formalize these ideas into an algorithm using the concept of \emph{classes}.
\begin{definition}[Class of an operation]
    For an operation $q$, we define its class as $k_q := \lfloor \log_2 p_q \rfloor$.
\end{definition}

\subsection{The Algorithm}
Our algorithm maintains the following objects.
\begin{itemize}
    \item 
We maintain a \emph{current class} $\curclass_j$ for every job.
When a job arrives, we set $\curclass_j \gets k_{j_1}$.
\item We maintain a queue $\afull$ of jobs, which is sorted by $\curclass_j$. Whenever a job arrives, we insert it into $\afull$. Let $\frontjob$ be the job with the smallest current class in $\afull$; break ties first in favor of the job with the smallest active operation, and all remaining ties arbitrarily. 
\item We maintain a stack $\apart$ of jobs. Let $\topjob$ be the job on top of $\apart$.
\end{itemize}
At any integer time $t$, we execute the following steps.
\begin{enumerate}
	\item
While $|\afull| \geq |J(t)| / 4$ \emph{and} the $\frontjob$ has 
a strictly smaller current class than $\topjob$, remove $\frontjob$ from $\afull$ and push it to $\apart$.
\item
Process $j =\topjob$ during $[t,t+1]$.
\item 
If 
$j$ completed, remove it from $\apart$. Otherwise, let $q$ be the active operation of $j$ at time $t+1$. 
If $k_{q} \geq \curclass_j + 1$, remove $j$ from $\apart$, update $\curclass_j \gets k_{q}$, and insert $j$ into $\afull$. 
\end{enumerate}

Note that in Step~3, $j$ will only be moved back to $\afull$ if $q$ just became active at time $t+1$, because otherwise the condition can never be true.
For the analysis, we use $\afull(t)$, $\apart(t)$, $\topjob(t)$, $\frontjob(t)$, and $\curclass_j(t)$ to denote the state of the algorithm at time $t$.
Also, we write $k(t) := \curclass_{\topjob(t)}(t)$ for the current class of the job that is processed during $[t,t+1]$.

\subsection{Chunks}
\label{sec:chunks}

We first analyze the structure of when the algorithm moves jobs between $\afull$ and $\apart$. This is exactly the structure of chunks of a job, which we formally define as follows.

\begin{definition}[Chunks]
    \label{def:chunk}
    Consider a job $j$ with $m$ operations $j_1,\ldots,j_m$. We partition the operations of $j$ into a sequence of \emph{chunks} $c_1,\ldots, c_{d_j}$ as follows:
    \begin{itemize}
        \item $c_1$ is the maximal prefix of operations such that all operations $j_i \in c_1$ are of class $\le k$, where $k$ is the class of $j_1$.
        \item For $h > 1$, the chunk $c_h$ is the maximal prefix of the operations $\{j_1,\ldots,j_m\} \setminus \bigcup_{h' < h} c_{h'}$ such that all operations $j_i \in c_h$ are of class $\le k$, where $k$ is the class of the first operation in $c_h$.
    \end{itemize}
    We use $p_c = \sum_{q \in c} p_q$ to refer to the processing time of a chunk.
    The \emph{class} $k_c$ of chunk $c$ is the maximum class of an operation in $c$, which is equivalent to the class of the first operation in $c$.
\end{definition}

See \Cref{fig:chunk-example} for an illustration of a job decomposed into chunks. 
The definition of chunks directly gives that the classes of the chunks of a job $j$ are strictly increasing.
\begin{fact}\label{fact:monotone-chunks}
    Let $c_1, \ldots, c_{d_j}$ be the chunks of a job $j$. Then, $k_{c_1} < \ldots < k_{c_{d_j}}$.
\end{fact}

In the schedule of the algorithm, let $c_j(t)$ denote the chunk that contains the active operation $o_j(t)$ of an active job $j$ at time $t$. 
We say that this is the \emph{active chunk} at time $t$. Let $r_c$ denote the first time when chunk $c$ is active, i.e., the time it becomes active.
A chunk $c$ is \emph{alive} at time $t$ if its job $j(c)$ is active at $t$. Similar to operations, each active chunk is alive but not every alive chunk is active.
We observe that the description of the algorithm ensures that chunks are exactly the entities of consecutive operations that the algorithm works on before moving a job back to $\afull$. The following observation follows from the definition of the algorithm, we defer the proof to~\Cref{app:chunks}.

\begin{restatable}{observation}{obsAlgChunks}
    Consider a job $j$ with $\ell$ chunks $c_1,\ldots,c_\ell$. The algorithm inserts $j$ into $\afull$ at time $t$ if and only if $t \in \{r_{c_1},\ldots,r_{c_\ell}\}$. Define $r_{c_\ell+1}= C_j$ for the completion time $C_j$ of job $j$. Each interval $I_i = [r_{c_i},r_{c_{i+1}}]$ with $i \in \{1,\ldots,\ell\}$ satisfies the following properties:
    \begin{enumerate}
        \item Job $j$ is moved to $\apart$ exactly once during $I_i$.
        \item $\curclass_j(t) = k_{c_i}$ for all $t \in [r_{c_i},r_{c_{i+1}})$.
    \end{enumerate}
\end{restatable}

\begin{figure}[t]
\centering
\begin{tikzpicture}[xscale=0.8,yscale=0.57]
  \tikzset{
    opbar/.style={draw=black!65, line width=0.35pt},
    firstop/.style={draw=black!85, line width=1.5pt},
    c1/.style={fill=job1},
    c2/.style={fill=job2},
    c3/.style={fill=job3}
  }

  \foreach \y in {0,1,...,6}{
    \draw[gray!15] (0,\y) -- (12,\y);
  }
  \draw[->] (0,0) -- (12.5,0) node[right] {operation $j_i$};
  \draw[->] (0,0) -- (0,6.2) node[above] {class $k_{j_i}$};

  \draw[c1,firstop] (0.1,0) rectangle (0.9,2);  
  \foreach \i/\h in {2/1,3/2}{
    \draw[c1,opbar] ({\i-0.9},0) rectangle ({\i-0.1},\h);
  }

  \draw[c2,firstop] (3.1,0) rectangle (3.9,3);  
  \foreach \i/\h in {5/2,6/3,7/1,8/2}{
    \draw[c2,opbar] ({\i-0.9},0) rectangle ({\i-0.1},\h);
  }

  \draw[c3,firstop] (8.1,0) rectangle (8.9,5);  
  \foreach \i/\h in {10/3,11/1,12/5}{
    \draw[c3,opbar] ({\i-0.9},0) rectangle ({\i-0.1},\h);
  }

  \foreach \i in {1,...,12}{
    \node[anchor=north] at ({\i-0.5},0) {$\i$};
  }

  \draw[densely dashed] (3,0) -- (3,6.0);
  \draw[densely dashed] (8,0) -- (8,6.0);

  \draw[decorate,decoration={brace,mirror,raise=3pt}] (0.1,-0.85) -- (2.9,-0.85)
    node[midway,below=7pt] {chunk $c_1$};
  \draw[decorate,decoration={brace,mirror,raise=3pt}] (3.1,-0.85) -- (7.9,-0.85)
    node[midway,below=7pt] {chunk $c_2$};
  \draw[decorate,decoration={brace,mirror,raise=3pt}] (8.1,-0.85) -- (11.9,-0.85)
    node[midway,below=7pt] {chunk $c_3$};

  \node[anchor=south] at (1.5,2) {$k_{c_1}=2$};
  \node[anchor=south] at (5.5,3) {$k_{c_2}=3$};
  \node[anchor=south] at (10,5) {$k_{c_3}=5$};
\end{tikzpicture}
\caption{Example chunk structure for a single job with $12$ operations and $3$ chunks. Each bar is an operation.}
\label{fig:chunk-example}
\end{figure}
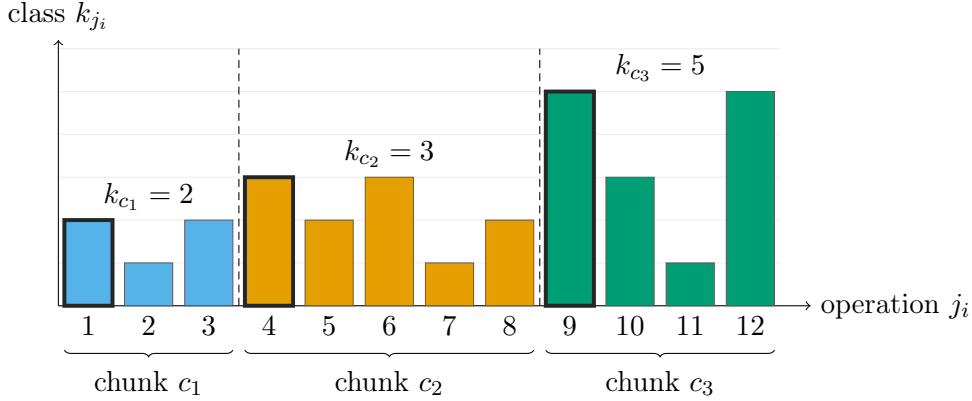

Since for every job $j \in J(t)$ there exists exactly one active chunk $c_j(t)$, we will without loss of generality refer to $J(t)$ as the set of active chunks at time $t$. Similarly, we use $\afull(t)$ and $\apart(t)$ for the set of active chunks of jobs in $\afull(t)$ and $\apart(t)$, respectively.

\subsection{Instance Parameters $\ops_1$ and $\ops_2$}

Based on the chunk structure of an instance, we derive two parameters $\ops_1$ and $\ops_2$ that
describe the maximal chunk size and the maximal number of chunks in a job, respectively.
Our analysis and performance guarantees will be based on these instance parameters.

\begin{definition}[$\ops_1$ and $\ops_2$]
    We define $\ops_1 := \max_{j \in J} d_j$ as the maximum number of chunks belonging to a single job, where $d_j$ is the number of chunks of job $j$,
    and we define $\ops_2 := \max_{c \in C} |c|$ as the size (number of operations) of the largest chunk.
\end{definition}

We can observe that the class of a chunk $c$ together with its length/size gives us an approximation of its processing time $p_c$. Hence, the class of a chunk can be used to approximate its size.

\begin{observation}\label{obs:chunk-size}
    Let $c$ be a chunk of class $k_c$ of some job $j$, then $p_c = \sum_{q \in c} p_{q} \le  \ops_2 \cdot 2^{k_c+1} \le 2 \cdot \ops_2 \cdot p_{q_1}$ where $q_1$ is the first operation in $c$.
\end{observation}

\begin{proof}
    Since $q_1$ is of class $k_c$, we have $p_{q_1}\ge 2^{k_c}$. All $q \in c \setminus \{q_1\}$ are of a class $\le k_c$ by~\Cref{def:chunk}. Hence, $p_q \le 2^{k_c+1} \le 2 \cdot p_{q_1}$ for all $q \in c \setminus \{q_1\}$. In conclusion, $\sum_{q \in c} p_{q} \le \min\{ \ops_2 \cdot 2^{k_c+1}, 2 \cdot \ops_2 \cdot p_{q_1}\}$.
\end{proof}

\subsection{Local Competitive Analysis and the Local Chunk LP}
\label{sec:local}

Fix a time~$\ts$. 
Our goal is to prove $|J(\ts)| \leq O(\ops_1 \cdot \ops_2) \cdot |J^*(\ts)|$, which implies local competitiveness.
Similarly to our analysis of Operations-SRPT in~\Cref{sec:operations-srpt}, we will show this bound indirectly by comparing $|J(\ts)|$ against the number of chunks that are alive in the optimal solution, denoted by $|J^*_c(\ts)|$. Since each job in $J^*(\ts)$ is composed of at most $\ops_1$ chunks, we conclude that
\begin{equation}
    |J^*_c(\ts)| \leq \ops_1 \cdot |J^*(\ts)|. \label{eq:opt_chunk_bound}
\end{equation}

Henceforth, our goal is to prove the following lemma, which will together with \eqref{eq:opt_chunk_bound} imply
\Cref{thm:main2}.

\begin{lemma}\label{lemma:local_chunk_bound}
	It holds that $|J(\ts)| \leq O(\ops_2) \cdot |J^*_c(\ts)|$.
\end{lemma}

Our proof relies on a linear programming relaxation for the problem of minimizing the number of alive chunks at time $\ts$. It follows the same structure as the LP in \cite{GuptaKPW26} for minimizing the number of alive jobs at time $\ts$.
Let $C$ denote the set of all chunks. 
We denote with $j(c)$ the job to which chunk $c$ belongs.
For a set $S \subseteq C$ of chunks, define the excess at time $\ts$ as $e(S) := \max(0, p(S) - (\ts- \ell_S))$ where $\ell_S = \min_{c \in S} r_{j(c)}$ and $p(S) = \sum_{c \in S} p_c$. 
For each chunk $c$, let $x_{c}$ be a variable indicating whether chunk $c$ is alive at time $\ts$.
While these quantities describe the state at time $\ts$, we do not add $\ts$ to the notation to keep it light.
The local chunk LP for time $\ts$ can be written as follows. 
\begin{alignat*}{3}
(\hypertarget{localchunklp}{\chunklp(\ts)}) \quad \min \; & \sum_{c \in C: \ts \geq r_{j(c)}} x_c \\
    \text{s.t.} \;  & \sum_{c \in S} \min(p_c, e(S)) \cdot x_c \geq e(S) \quad&& \forall S \subseteq C \\
    & x_c \geq 0 \quad&& \forall c \in C
\end{alignat*}

Intuitively, the LP operates on a surrogate instance $J_c$ that 
treats each chunk $c \in C$ as an individual job with processing 
time $p_c$ and release date $r_{j(c)}$. That is, all chunks $c$ that 
belong to job $j(c)$ are released at the same time and can be processed 
in any order. For the objective function of minimizing the number of alive 
chunks at time $\ts$, this is clearly a relaxation. This is a similar 
idea to the analysis of Operations-SRPT in~\Cref{sec:operations-srpt}. 
It is easy to verify that the covering constraints of the LP are valid 
constraints for the problem of minimizing the number of alive jobs at time 
$\ts$, which implies the following lemma, whose proof is deferred to~\Cref{app:chunk-lp}.

\begin{restatable}{lemma}{factLPRelax}
\label{fact:lp}
The optimal objective value of (\hyperlink{localchunklp}{$\chunklp(\ts)$}) is at most $|J^*_c(\ts)|$.
\end{restatable}

In order to compare $|J(\ts)|$ to the optimal objective value of (\hyperlink{localchunklp}{$\chunklp(\ts)$}), we use a dual fitting analysis. 
The dual of (\hyperlink{localchunklp}{$\chunklp(\ts)$}) with variables $y_S$ can be written as follows:
\begin{alignat*}{3}
(\hypertarget{localchunkdp}{\chunkdp(\ts)}) \quad \max \;  &  \sum_{S} e(S) \cdot y_{S} \\
    \text{s.t.} \;  & \sum_{S: c \in S} \min(p_c, e(S)) \cdot y_{S} \leq 1 \quad&& \forall c \in C \colon r_{j(c)} \le \ts \\
    & y_{S} \geq 0 \quad&& \forall S \subseteq C
\end{alignat*}

\subsection{Dual Fitting}

To prove~\Cref{lemma:local_chunk_bound} via dual fitting, we 
construct a feasible solution to (\hyperlink{localchunkdp}{$\chunkdp(\ts)$}) with an objective value of $\Omega(|J(\ts)| / \ops_2)$. 
Therein, a crucial step is identifying subsets $S \subseteq C$ of chunks with a large excess $e(S)$, as the corresponding dual variable $y_S$ 
may have a large contribution to the objective. 
In~\Cref{sec:excess-lemmas}, we identify such sets and give lower bounds on their excess. Then, 
we construct a dual solution, 
lower bound its objective value by $\Omega(|J(\ts)| / \ops_2)$, and prove its feasibility. 

\subsubsection{Excess Lemmas}
\label{sec:excess-lemmas}

We start by proving lower bounds on the excess of certain chunk sets at time $\ts$. 
A natural strategy, which is also used in~\cite{GuptaKPW26}, for proving excess bounds on a \emph{job level}, is to identify a set of jobs $J'$ and an interval $I = [\ell,\ts]$ such that (i) all jobs in $J'$ are released during $I$, (ii) the algorithm only works on jobs in $J'$ during $I$, and (iii) the set $J'$ has a significant remaining volume $p_{J'}(\ts) := \sum_{j \in J'} p_j(\ts)$. 
Such $J'$ and $I$ then imply for the set of \emph{jobs} $J'$ an excess lower bound $e(J') := p(J') - (\ts-\ell) \ge p_{J'}(\ts)$.

To prove excess lower bounds for \emph{sets of chunks} $S$, we would like to follow this strategy but are facing an additional challenge: The algorithm treats each chunk $c \in S$ as its own job with release date $r_c$. However, to bound the excess $e(S)$ with the strategy above, we need an interval $I$ such that all \emph{jobs} $j(c)$ with $c \in S$ are released during $I$. Hence, we have to handle the disconnect between $r_c$ and $r_{j(c)}$. In the following, we provide proofs for excess bounds that take this challenge into account.

To this end, let $F(k) \subseteq J(\ts)$ denote the set of full chunks of class $k$ active at time $\ts$. Call class $k$ \emph{crucial} if $F(k) \neq \emptyset$.

\begin{definition}[$t_{\geq k}$, $S_{<k}$]
\label{def:geqk}
Let $t_{\geq k}$ be the last time before $\ts$ when the algorithm processed a chunk of class $\geq k$. Let $S_{<k}$ denote the set of chunks $c$ with the following properties:
\begin{enumerate}
    \item $c$ is of a class $< k$ and
    \item $j(c)$, the job to which $c$ belongs, is released during $[t_{\geq k}+1,\ts]$.
\end{enumerate}    
\end{definition}
We prove our excess bounds for the sets $S_{<k}$ under the following assumption.

\begin{restatable}{assumption}{asmReduced}\label{asm:reduced-property}
For each $k$, assume that every chunk $c \in J(\ts)$ that becomes active after $t_{\geq k}$ is of class $<k$.
\end{restatable}

In~\Cref{sec:reduced-instance}, we show that this is without loss of generality for our goal of proving~\Cref{lemma:local_chunk_bound} by exploiting the concept of a reduced instance as introduced in~\cite{GuptaKPW26}. 

We also rely on the following lemma, which formulates a crucial property of the algorithm regarding the chunk $c$ that is processed at a time $t$. In particular, such a chunk has the smallest class among all active chunks with possibly a single exception. We defer the proof to~\Cref{app:properties}.

\begin{restatable}{lemma}{lemStrictClasses}\label{lem:strict-classes}
	Let $c$ be the chunk of class $k(t)$ processed during $[t,t+1]$. 
	\begin{enumerate}
		\item If there exists a chunk $c'$ of class $< k(t)$ in $J(t)$, then there cannot be another chunk of class $\leq k(t)$ in $J(t)\setminus\{c,c'\}$.
		\item If there exists another chunk $c' \neq c$ of class $k(t)$ in $J(t)$, then there cannot be any chunk of class $<k(t)$ in $J(t)$.
	\end{enumerate}
\end{restatable}

With~\Cref{asm:reduced-property} and~\Cref{lem:strict-classes} in place, we are ready to prove excess lower bounds.

\begin{lemma}[Excess Lemma 1]\label{lem:excess-1}
    If class $k$ is crucial, then
	$e(S_{<k}) \geq \sum_{k'< k} \sum_{c \in F(k')} p_{c}$.
\end{lemma}

\begin{proof}
    We first show that at time $t_{\geq k}$, there cannot exist 
    any active chunk of class $<k$.  
    To this end, we first argue that there is at least 
    one full chunk of class $k$ at time $t_{\geq k}$.
    Let $c \in F(k)$, which exists because class $k$ is crucial.
	If~$c$ becomes active at or after time $t_{\geq k}+1$, then it must be of class $<k$ by~\Cref{asm:reduced-property}, a contradiction to $c \in F(k)$.
		We conclude that every chunk in $F(k)$ became active before $t_{\geq k}$, and there exists at least one such chunk $c \in F(k)$.
        By the definition of time $t_{\geq k}$, 
          the algorithm must have processed a chunk $c'$ 
          of class $\geq k$ during $[t_{\geq k},t_{\geq k}+1]$.
		Since $c$ is full at time $\ts$, $c' \neq c$, and thus, \Cref{lem:strict-classes} implies that there cannot exist a chunk of class $<k$ at time $t_{\geq k}$.		
        
	Now, let $c$ be a chunk of class $k'<k$ that becomes active during $[t_{\geq k}+1,\ts]$. 
    We next argue that the job $j=j(c)$ corresponding 
    to $c$ must have been released after $t_{\geq k}+1$. 
	To see this, assume that $j$ was released before $t_{\geq k}$ and hence an earlier chunk $c'$ of $j$ is active at time $t_{\geq k}$. 
	By \Cref{fact:monotone-chunks}, $c'$ must also be of class $<k$. This is a contradiction, as we have established that at time $t_{\geq k}$ there cannot exist any active chunk of class $<k$.
	Hence $j$ must be released after $t_{\geq k}+1$.
	
    Thus, all chunks $c$ of class $<k$ that became active during $[t_{\geq k}+1,\ts]$ belong to jobs $j(c)$ that are released during $[t_{\geq k}+1,\ts]$. 
	By definition, all of these chunks belong to $S_{<k}$.
	Moreover, by the definition of time $t_{\geq k}$, the algorithm only works on those chunks during $[t_{\geq k}+1,\ts]$. 
	By the definition of excess,
	$e(S_{<k}) = \max\{0,p(S_{<k}) - (\ts - (t_{\geq k}+1))\} $.
	Since chunks $c \in F(k')$ for $k'<k$ are full at time $\ts$ and belong to $S_{<k}$,  
	we conclude that 
	\[
	    e(S_{<k}) \geq \sum_{k'< k} \sum_{c \in F(k')} p_{c} \ ,
	\]
	which completes the proof.     
\end{proof}

\begin{definition}[$t_{>k}, S_{\le k}$]
Let $t_{>k}$ be the last time before $\ts$ when the algorithm processed a chunk of class $> k$.
Let $S_{\le k}$ denote the set of chunks $c$ with the following properties.
\begin{enumerate}
    \item $c$ is of a class $\le k$, and
    \item $j(c)$, the job to which $c$ belongs, is released during $[t_{>k}+1,\ts]$.
\end{enumerate}
\end{definition}

\begin{lemma}[Excess Lemma 2]\label{lem:excess-2}
If class $k$ is crucial, then
$e(S_{\le k}) \geq \sum_{k' \leq k} \sum_{c \in F(k')} p_{c} - 2 \cdot  \ops_2 \cdot 2^{k+1}$. 
\end{lemma}
\begin{proof}
	We first establish that at time $t_{>k}$ there can be at 
	most one active chunk of class $\leq k$. 
	Let $c$ be such a chunk belonging to a job $j = j(c)$ (if it exists; otherwise the claim is trivial).
	By definition of $t_{>k}$, the algorithm processes a chunk $c'$ of class $>k$ during $[t_{>k},t_{>k}+1]$. This immediately implies that $c$ is full at $t_{>k}$.
	By \Cref{lem:strict-classes} and the existence of $c$, there cannot exist any additional active chunk of class $<k$ at time $t_{>k}$. Hence, $c$ is the only (full or partial) active chunk of class $\leq k$ at time $t_{>k}$.
	
	We next claim that all chunks of class $\leq k$ that do not belong to job $j$ and that became active during $[t_{>k}+1,\ts]$ must belong to jobs released after $t_{>k}$.
	To see this, let $c'$ be a chunk of class $\leq k$ that became active during $[t_{>k}+1,\ts]$ that belongs to a job $j'$ released before $t_{>k}+1$. 
	Thus, there must exist an earlier chunk $c''$ of $j'$ that is active at time $t_{>k}$. 
	By \Cref{fact:monotone-chunks}, $c''$ must also be of class $\leq k$. 
	This is a contradiction, since $c$ is the only such chunk at time $t_{>k}$.
	Thus, except chunks of job $j$, all chunks $c$ of class $\leq k$ that became active during $[t_{>k}+1,\ts]$ belong to jobs $j(c)$ that are released during $[t_{>k}+1,\ts]$. 
	By definition, all of these chunks belong to $S_{\le k}$.

    We have established that there can be at most one chunk $c$ of class $\leq k$ that is full at time $t_{>k}$, which is not part of $S_{\le k}$. Let $j = j(c)$ denote the job of this chunk.
    By choice of $t_{>k}$, the algorithm only works on chunks in $S_{\le k}$ and on chunks belonging to $j$ during $[t_{>k}+1,\ts]$. Let $\Delta$ denote the amount of time during $[t_{>k}+1,\ts]$ which the algorithm spends working on $j$. Then, by definition of excess,
	\[
	    e(S_{\le k}) + \Delta \geq \sum_{k'\leq k} \sum_{c \in F(k')} p_{c} \ .
	\]
    It remains to bound $\Delta$. Note that whenever the algorithm works on job $j$ during $[t_{>k}+1,\ts]$, then it either works on the chunk $c$ of class $\leq k$ that was full at time $t_{>k}$ or on later chunks that become active due to the completion of $c$. Let $c_1 \le \ldots \le c_d$ denote these chunks indexed in the order they are processed. Since all of these chunks are processed during $[t_{>k}+1,\ts]$, they have to be of class $\leq k$. Furthermore, by~\Cref{fact:monotone-chunks}, the classes of these chunks are strictly increasing.
    Since each $c_i$ is composed of at most $\ops_2$ operations, the total volume of $c_i$ is at most $\ops_2 \cdot 2^{k_{c_i}+1}$. We bound $\Delta$ by summing over the volume of all these $c_i$'s:
    $$
    \Delta \le \sum_{i = 1}^d \ops_2 \cdot 2^{k_{c_i}+1} \le \ops_2 \cdot \sum_{i=0}^k 2^{i+1} \le 2 \cdot \ops_2 \cdot 2^{k+1}. 
    $$
    We can conclude with the excess bound
    \[
	    e(S_{\le k}) + 2 \cdot \ops_2 \cdot 2^{k+1} \geq \sum_{k'\leq k} \sum_{c \in F(k')} p_{c} \ ,
	\]    
	which concludes the proof.
\end{proof}

\subsubsection{Construction and Analysis of the Dual Solution}

We construct our dual solution based on the sets $S_{<k}$ and $S_{\le k}$ as defined in the previous section.
For every crucial class $k$, we define
\begin{itemize}
	\item $y_{S_{<k}} := \frac{|F(k)|}{3\ops_2 \cdot e(S_{<k})}$ if $|F(k)| < 6\ops_2$,
	\item $y_{S_{\le k}} := \frac{1}{\ops_2 2^k}$ if $|F(k)| \geq 6\ops_2$,
\end{itemize}
and for all other sets $S$ we define $y_S := 0$.

\paragraph{Dual Objective Value.} Next, we lower bound the objective value of the dual solution in terms of $\frac{J(\ts)}{\ops_2}$ by proving the following lemma.

\begin{lemma}
\label{lem:dual:objective}
    We have $\sum_{S} e(S) y_{S} 
    \geq \frac{|J(\ts)|}{12\ops_2} - \frac{1}{3\ops_2}$.
\end{lemma}

\begin{proof}
	For every class $k$ with $|F(k)| < 6\ops_2$ we get a contribution equal to
    \[
	    e(S_{<k}) y_{S_{<k}} = e(S_{<k}) \cdot \frac{|F(k)|}{3\ops_2 \cdot e(S_{<k})} =  \frac{1}{3\ops_2} |F(k)|.
    \]
	
	For every class $k$ with $|F(k)| \geq 6\ops_2$,
    we have 
	\begin{align*}
	e(S_{\le k}) y_{S_{\le k}} 
	&\geq \bigg( \sum_{k'\le k}\sum_{c \in F(k')} p_{c} - 2 \cdot \ops_2 \cdot 2^{k+1} \bigg) \cdot \frac{1}{\ops_2 2^k} 
    \geq \bigg(\sum_{c \in F(k)} p_{c} - 2 \cdot \ops_2 \cdot 2^{k+1} \bigg) \cdot \frac{1}{\ops_2 2^k} \\ 
	&\geq \frac{2^k}{\ops_2 2^k} |F(k)| - 4 
    \geq \frac{1}{\ops_2} |F(k)| - 4 
    \geq \frac{1}{\ops_2} |F(k)| - \frac{4}{6\ops_2} |F(k)| 
    = \frac{1}{3\ops_2} |F(k)| \ .
	\end{align*}
    Here, the first inequality uses~\Cref{lem:excess-2}.	
    
	Summing over all crucial classes, we obtain
	$\sum_{S} e(S) y_{S} \geq \frac{|\afull(\ts)|}{3\ops_2}$.
    The final inequality uses that $|\afull(\ts)| \geq \frac{|J(\ts)|}{4} - 1$, which is a property of the algorithm (cf.~\Cref{lem:alg:property} in~\Cref{app:properties} or~\cite{GuptaKPW26}) for a proof).
\end{proof}

\paragraph{Dual Feasibility.} Finally, we argue about the feasibility of the dual solution.
Fix a chunk $c^*$ of class $k^*$ with $r_{j(c^*)} \le \ts$. 
Our goal is to show that our dual variables violate the constraint of $c^*$ only by some constant factor, which then implies that scaling our variables by that factor yields a feasible dual solution.

By definition of our dual variables, 
the only sets $S$ that can have $c^* \in S$ and $y_S > 0$ are sets $S_{\le k}$ and $S_{<k}$ with $k \ge k^*$. 
Let 
\begin{itemize}
	\item $K_0 = \{k \geq k^* \mid c^* \in S_{<k} \}$ and 
	\item $K_1 = \{k \geq k^* \mid c^* \in S_{\le k} \}$.
\end{itemize}

First, we analyze the contribution of sets $S_{\le k}$ with $k \in K_1$:

\begin{lemma}
    \label{claim:dual:feasibility:1}
    $\sum_{k \in K_1} \min(p_{c^*}, e(S_{\le k})) \cdot y_{S_{\le k}} < 4$.
\end{lemma}

\begin{proof}
For every $k \in K_1$, we have $y_{S_{\le k}} = \frac{1}{\ops_2 2^k}$ if $|F(k)|\geq 6\ops_2$ and  $y_{S_{\le k}} = 0$ otherwise.
Since $c^*$ is of class $k^*$, we have $p_{c^*} \leq \ops_2 2^{k^*+1}$. Thus,
\[
\sum_{k \in K_1} \min(p_{c^*}, e(S_{\le k})) \cdot y_{S_{\le k}} 
\leq \sum_{k \geq k^*} p_{c^*} \cdot \frac{1}{\ops_2 2^k}
\leq \sum_{k \geq k^*} \frac{\ops_2 2^{k^*+1}}{\ops_2 2^k} < 4 \ .
\]   
\end{proof}

Next, we analyze the contribution of sets $S_{<k}$ with $k \in K_0$. 

\begin{lemma}
        \label{claim:dual:feasibility:2}
        $\sum_{k \in K_0} \min(p_{c^*}, e(S_{<k})) \cdot y_{S_{<k}} \leq 10$.
\end{lemma}

\begin{proof}

To this end, for all $0 \leq \ell \leq \Delta := \lfloor \log_2(3\ops_2) \rfloor$, let $T_\ell = \{k \in K_0 \mid 2^\ell \leq |F(k)| < 2^{\ell+1} \}$.
For each $k \in T_\ell$, by definition $y_{S_{<k}} = \frac{|F(k)|}{3\ops_2 \cdot e(S_{<k})}$ if $|F(k)| < 6\ops_2$ and $y_{S_{<k}} = 0$ otherwise. Thus,
\begin{align}
	\min(e(S_{<k}), p_{c^*}) \cdot y_{S_{<k}} 
	&= \min(e(S_{<k}), p_{c^*}) \cdot \frac{|F(k)|}{3\ops_2 \cdot e(S_{<k})} \notag \\
	&\leq \min(e(S_{<k}), p_{c^*}) \cdot \frac{2^{\ell+1}}{3\ops_2 \cdot e(S_{<k})} 
	= \frac{2^{\ell+1}}{3\ops_2} \min \bigg(1, \frac{p_{c^*}}{e(S_{<k})} \bigg) \ . \label{eq:dual:feasibility:2}
\end{align}
Now, assume that $T_\ell \neq \emptyset$ and let $k_1 < k_2 < \ldots < k_{|T_\ell|}$ denote the classes of $T_\ell$. All those classes are crucial. Hence, we can apply~\Cref{lem:excess-1} to all of the corresponding sets $S_{< k_i}$. For each $i \geq 2$, this implies:
\begin{align}
    e(S_{< k_i}) \geq \sum_{k' < k_i} \sum_{c \in F(k')} p_{c}
    \geq \sum_{c \in F(k_{i-1})} p_{c} 
    \geq |F(k_{i-1})| \cdot 2^{k_{i-1}} \geq 2^\ell 2^{k_{i-1}} \ . 
    \label{eq:dual:feasibility:3}
\end{align}
Hence,
\begin{align*}      
\sum_{k \in T_\ell} \min(e(S_{<k}), p_{c^*}) \cdot y_{S_{<k}} 
    &\leq \frac{2^{\ell+1}}{3\ops_2}  \sum_{k \in T_\ell} \min \bigg(1, \frac{p_{c^*}}{e(S_{<k})} \bigg) 
    \leq \frac{2^{\ell+1}}{3\ops_2}  \bigg( 1 + \frac{1}{2^\ell}\sum_{i \geq 2} \frac{p_{c^*}}{2^{k_{i-1}}} \bigg) \\
    &\leq \frac{2^{\ell+1}}{3\ops_2} \bigg(1  + \frac{\ops_2 2^{k^*+1}}{2^\ell} \sum_{i \geq 2} \frac{1}{2^{k_{i-1}}} \bigg) 
    \leq \frac{2^{\ell+1}}{3\ops_2} \bigg(1  + \frac{\ops_2 2^{k^*+1}}{2^\ell} \frac{2}{2^{k_1}} \bigg) \ ,
\end{align*}
where the first inequality uses \eqref{eq:dual:feasibility:2},
the second inequality uses \eqref{eq:dual:feasibility:3} and that $\min(1, p_{c^*} / e(S_{< k_1})) \leq 1$,
the third inequality uses that $p_{c^*} \leq \ops_2 2^{k^*+1}$, and
the fourth inequality uses $\sum_{i \geq 2} \frac{1}{2^{k_{i-1}}} \leq 2 \cdot 2^{-k_1}$.

For each $\ell$ with $T_\ell \neq \emptyset$, let $k(\ell)$ denote its smallest class.
Since $y_{S_{<k}} = 0$ for all $k$ with $|F(k)| \geq 6\ops_2$, we can finally sum over all indices $\ell$ with $T_\ell \neq \emptyset$ to verify the dual constraint
\begin{align*}
    \sum_{k \in K_0} \min(e(S_{<k}), p_{c^*}) \cdot y_{S_{<k}}
    &= \sum_{\ell=0}^{\Delta} \sum_{k \in T_\ell} \min(e(S_{<k}), p_{c^*}) \cdot y_{S_{<k}} \\
    &\leq \sum_{\substack{\ell=0\\ T_\ell \neq \emptyset}}^{\Delta} \frac{2^{\ell+1}}{3\ops_2} \bigg(1  + \frac{\ops_2 2^{k^*+1}}{2^\ell} \frac{2}{2^{k(\ell)}} \bigg)  
    = \sum_{\substack{\ell=0\\ T_\ell \neq \emptyset}}^{\Delta} \frac{2^{\ell+1}}{3\ops_2} + \sum_{\substack{\ell=0\\ T_\ell \neq \emptyset}}^{\Delta}  \frac{2 \cdot 2^{k^*+1}}{3} \frac{2}{2^{k(\ell)}} \ .
\end{align*}
For the first sum, we have $\sum_{\ell=0}^{\Delta}\frac{2^{\ell+1}}{3\ops_2} \leq \frac{4}{3\ops_2} 2^{\log_2(3\ops_2)} = 4$.
For the second sum, since $T_0, \ldots, T_\Delta$ partition all crucial classes $\geq k^*$, all $k(0),\ldots,k(\Delta)$ are pairwise distinct and at least $k^*$. Hence, we get $\sum_{\ell=0}^\Delta 2^{-k(\ell)} \leq 2 / 2^{k^*}$. Thus, the second sum is at most $16/3$. 
Since $4 + 16/3 < 10$, this completes the proof.
\end{proof}

\paragraph{Putting Everything Together.}
We can finally prove \Cref{thm:main2}, which then implies \Cref{thm:main1}.
\begin{proof}[Proof of \Cref{thm:main2}]
Combining~\Cref{claim:dual:feasibility:1} and~\Cref{claim:dual:feasibility:2}, we get
\[
\sum_{S: c^* \in S} \min(p_{c^*}, e(S)) \cdot y_{S} \leq 14 \ .
\]
Thus, the dual assignment $y_{S} / 14$ for all $S \subseteq C$ is a feasible solution to (\hyperlink{localchunkdp}{$\chunkdp(\ts)$}). 
By~\Cref{lem:dual:objective}, this dual solution has an objective value of at least $\frac{|J(\ts)|}{168\ops_2}- \frac{1}{52\ops_2}$. Hence, \Cref{fact:lp} and weak duality give
$$
|J(\ts)| \le 168\ops_2 \cdot |J^*_{c}(\ts)| +\frac{1}{52\ops_2} \ ,
$$
where $|J^*_{c}(\ts)|$ is the number of active chunks in $J^*$ at time $\ts$. This concludes the proof of~\Cref{lemma:local_chunk_bound} and, using Inequality~\eqref{eq:opt_chunk_bound}, implies $|J(t)| \in \mathcal{O}(\ops_1 \cdot \ops_2) \cdot |J^*(t)|$ for any time $t$. Integrating over time then implies~\Cref{thm:main2}.
\end{proof}

\printbibliography

\appendix

\section{Further Related Work}
\label{app:related}

For non-clairvoyant scheduling on a single machine, Becchetti and Leonardi~\cite{BecchettiL04} presented an $O(\log n)$-competitive randomized algorithm, matching the lower bound of \cite{MotwaniPT94}; their deterministic lower bound of $\Omega(n^{1/3})$ remained unmatched until today.
Bender et al.~\cite{BenderMR04} proposed a model where $\lfloor \log_2 p_j \rfloor$ is revealed to an algorithm when job $j$ arrives, for which Becchetti et al.~\cite{BecchettiLMP04} gave an $O(1)$-competitive algorithm.
Building on these results, Azar, Leonardi, and Touitou assumed that a job $j$ arrives with a predicted processing time $\hp_j$ with prediction error $\mu = (\max_j \hp_j / p_j) \cdot (\max_j p_j / \hp_j)$, and presented algorithms with competitive ratios $O(\mu^2)$~\cite{AzarLT21} and $O(\mu \log \mu)$~\cite{AzarLT22}. Gupta et al.~\cite{GuptaKPW26} very recently improved this to $O(\mu)$,
which is best-possible~\cite{AzarLT22}.
Under $(1+\varepsilon)$-speed augmentation, a simple non-clairvoyant algorithm can be shown to be $O(1/\varepsilon)$-competitive~\cite{KalyanasundaramP00}.

The weighted objective is in terms of guarantees much harder, even in the 
clairvoyant setting: Bansal and Chan~\cite{BansalC09} showed that no online algorithm can be $O(1)$-competitive.
The currently best-known algorithms are due to Azar and Touitou~\cite{AzarT18}, building up on earlier works~\cite{BansalD07,ChekuriKZ01}. Under $(1+\varepsilon)$-speed augmentation, there are several $O(1)$-competitive algorithms~\cite{BansalD07,BecchettiLMP06}.
In the offline setting, unlike the unweighted problem, it is NP-hard~\cite{LabetoulleLLR84}, and a PTAS for minimizing the weighted flow time has been established only very recently~\cite{ArmbrusterRW23}.

Our problem can be considered a special case of scheduling with \emph{online precedence constraints} and \emph{weighted} jobs. In scheduling with precedence constraints, we are given a directed acyclic graph $G = (J,E)$ that uses the jobs as vertices and formulates precedence constraints between the jobs, i.e., a job $j$ can only be processed once all $j'$ with $(j',j) \in E$ have been completed. In the online variant, the job set and precedence constraint graph are initially unknown to the scheduler, and a job is only revealed to the scheduler once all predecessors have been completed. This scheduling model has for example been studied for makespan~\cite{AzarE02} and for weighted total completion time minimization, where no deterministic algorithm has a competitive ratio better than $\Omega(n)$ but $\mathcal{O}(1)$-competitive algorithms are possible in learning-augmented settings~\cite{LassotaLMS23}. The model captures our \ofts: A job $j$ with $m$ operations in our model can be  represented by introducing (i) a single job $j_i'$ for each operation $j_i$ of $j$, (ii) introducing precedence constraints $(j'_i,j'_{i+1})$ for all $1 \le i < m$, and (iii) using weights $w_{j'_m}=1$ and $w_{j'_i} = 0$ for all $i < m$. To the best of our knowledge, online precedence constraints have not yet been studied for the total flow time objective.

As outlined in the introduction, our model is related to non-clairvoyant flow time minimization with untrusted predictions~\cite{GuptaKPW26,AzarLT21,AzarLT22}. Untrusted predictions have also been studied for the special case of total completion time minimization (see, e.g.,~\cite{PurohitSK18,WeiZ20,,ImKQP23, BenomarP23,BenomarP24nonclarivoyant, DinitzILMV22, EliasKMM24,LindermayrM25permutation,BenomarCLS25}).

\section{Lower Bounds}

In this section, we prove several lower bounds for the \ofts. For the sake of convenience, we state our lower bound constructions using rational release dates and processing times. Lower bound instances with integer release dates and processing times can then be obtained via scaling.

For a fixed algorithm and a fixed instance, we denote by $\delta(t)$ the number of unfinished jobs in the algorithm's schedule and by $\delta^*(t)$ the number of unfinished jobs at time $t$ in an optimal solution.
To prove lower bounds, we use the following well-known technique, see e.g.~\cite{AzarLT22,MotwaniPT94,GuptaKLSY25}.

\begin{proposition}\label{prop:dos}
    If for every sufficiently large integer $N$ there exists an instance such that $\delta^*(t) = N$ and $\delta(t) \geq \rho \cdot \delta^*(t)$ for some constant $\rho$ at some time $t$, then the algorithm has a competitive ratio of at least $\rho$.
\end{proposition}

\subsection{A Lower Bound for Deterministic Algorithms}

We first present a lower bound for deterministic algorithms.
In fact, our lower bound follows directly from a lower bound of Motwani, Philipps, and Torng \cite{MotwaniPT94}. They show that every deterministic non-clairvoyant algorithm has a competitive ratio of at least $P$, where $P$ is the ratio between the largest and smallest processing time. The construction works for jobs with operation sizes $0$ and $1$ where the smallest job has processing time $1$ and the largest jobs has a processing time of $m$.
Since this implication is immediate, we give the proof of \cite{MotwaniPT94} in the language of operations.

\begin{theorem}\label{thm:det-lb}
  For any integer $\ops \geq 2$, the competitive ratio of any deterministic algorithm for minimizing the total flow time with $\ops$ operations is at least~$\ops$, even if all operations have size $0$ or $1$ and are monotone non-increasing.
\end{theorem}

\begin{proof}
Let $N$ be a larger integer. At time $0$, we release a set $J$ of $N(m+1)$ jobs with $p_{j_1} = 1$ for all $j \in J$.
Whenever the algorithm completes an operation, the next operation again has size $1$.
We assume w.l.o.g.\ by the integrality of the instance that the algorithm only preempts at integer times.
Let $s$ be the time when the algorithm completed the $(m-1)$th operation of $N$ jobs. Let $J_1$ denote the set of those jobs.
At this time, we set the processing time of all hidden operations to $0$. 
Let $t$ be the time when the algorithm completed exactly $m$ jobs, which must be after time $s$.
By construction, $\delta(t) = mN$.
Since at time $s$ all jobs have the same remaining total length, we can assume without loss of generality that the algorithm completed all $m$ jobs in $J_1$ by time $t$.
Hence, it worked for $t - Nm$ units on jobs in $J \setminus J_1$ until time $t$, 
and at time $t$, each of the $Nm$ jobs in $J \setminus J_1$ has remaining length $1$.
Thus, another solution can finish all jobs in $J \setminus J_1$ until time $t$,
and thus,
$\delta^*(t) \leq N$ because only the $N$ jobs in $J_1$ remain.
Then, the lemma follows from \Cref{prop:dos}.
\end{proof}

\subsection{A Lower Bound for Randomized Algorithms}

We next move to randomized algorithms. We show that randomization does not significantly help to improve the competitive ratio over deterministic algorithms. We emphasize that this lower bound in new and does not appear in \cite{MotwaniPT94}. Our proof roughly follows the same framework as the randomized lower bound proofs for related flow time problems in~\cite{AzarLT22,GuptaKLSY25} but tailors it to scheduling with $m$ operations.

\begin{theorem}\label{theorem:randomized-lb}
    The competitive ratio of any randomized algorithm against an oblivious adversary for minimizing the total flow time with $\ops$ operations is at least~$\Omega(\ops)$.
\end{theorem}

We later apply Yao's principle, so fix any deterministic algorithm.
We consider a distribution of instances where at time $0$ we release $n =  \floor{2^{\ops / 2}}$ jobs
with integer processing times $P_j$ drawn independently from the geometric distribution with mean $2$ at time $0$. That is, $P_j = p$ with probability $2^{-p}$ for every integer $p \geq 1$.
Note that $\EX[P_j] = 2$.
For every job $j$, we define its operation processing times as follows: $P_{j_\ell} = 1$ for all $1 \leq \ell \leq \min\{\ops,P_j\}$, $P_{j_\ell} = 0$ for all $P_j + 1 \leq \ell < \ops$, and $P_{j_\ops} = \max\{P_j - \ops + 1, 0\}$.

Let $t = \floor{2(n-n^{3/4})}$.
We first upper bound the expected number of alive jobs at time $t$ in the optimal solution.

\begin{lemma}\label{lemma:randomized-lb-opt}
  It holds that $\EX[\delta^*(t)] \leq \cO(\frac{n^{3/4}}{\log n})$.
\end{lemma}
\begin{proof}
  First, note that the total processing time $P$ of all jobs is a sum of $n$ independent random variables with mean $2n$ and variance $2n$. Thus, Chebyshev's inequality gives
  \[
      \pr[P \geq 2n + n^{3/4}] = \pr[P \geq 2n + n^{1/2} \cdot n^{1/4}] \leq \cO \left( \frac{1}{n^{1/2}} \right) \ .
  \]

  Moreover, let $b = \frac{\log n}{4}$ and let $B$ be the number of jobs with processing time more than $b$.
  Note that $B$ is binomially distributed, because it can be written as the sum of $n$ binary random variables $\ind[P_j > b]$ with equal success probabilities $\pr[P_j > b] = 2^{-b} = n^{-1/4}$.
  Therefore, $\EX[B] = n^{-1/4} \cdot n = n^{3/4}$ and $\VAR[B] = n^{-1/4} (1-n^{-1/4}) n = O(n^{3/4})$.
  Thus, Chebyshev's inequality gives
  \[
      \pr   [B \leq \tfrac{1}{2}n^{3/4} ]
      = \pr     [B - \EX[B] \leq -\tfrac{1}{2}\EX[B] ]
      \leq \frac{\VAR[B]}{(\frac{1}{2}\EX[B])^2}
      \leq \cO \left( \frac{1}{n^{3/4}} \right) \ .
  \]

  Let $\cE_1$ denote the event $P < 2n + n^{3/4}$ and let $\cE_2$ denote the event $B > \frac{1}{2}n^{3/4}$.
  By the union bound, we have
  \[
      \pr[\cE_1 \cap \cE_2] \ge 1 - \pr[\bar{\cE}_1 ] - \pr[\bar{\cE}_2] \ge 1 - \cO\left(\frac{1}{n^{1/2}}\right) - \cO\left(\frac{1}{n^{3/4}}\right)
      = 1 - \cO\left(\frac{1}{n^{1/2}}\right) \ .
  \]

  Conditioned on $\cE_1 \cap \cE_2$, we compute the maximum number of jobs of size at least $b$ that an optimum solution cannot complete until time $t$ as follows
  \begin{align*}
      \frac{1}{b} \left(P - t \right) \leq \frac{1}{b} \left(2n + n^{3/4} - 2(n-n^{3/4}) + 1 \right)
      \leq \cO\bigg( \frac{n^{3/4}}{\log n} \bigg) \ .
  \end{align*}
  Note that for sufficiently large $n$ it holds that $\cO(\frac{n^{3/4}}{\log n}) \leq \frac{1}{2}n^{3/4} \leq B$, i.e., there are many such long jobs.
  In total, the expected number of alive jobs in an optimum solution at time $t$ is at most
  \[
      \EX[\delta^*(t)] \leq \left(1-\cO\left(\frac{1}{n^{1/2}} \right) \right) \cdot \cO\bigg( \frac{n^{3/4}}{\log n} \bigg) + \cO \bigg( \frac{1}{n^{1/2}} \bigg) \cdot n = \cO\bigg( \frac{n^{3/4}}{\log n} \bigg) \ ,
  \]
  which concludes the proof of the lemma.
  \end{proof}

\begin{lemma}\label{lemma:randomized-lb-alg}
   $\EX[\delta(t)] \geq \Omega(n^{3/4})$.
\end{lemma}

\begin{proof}
We call a job $j$ \emph{short} if it has
$P_j \leq m$; otherwise \emph{long}. 
Let $\cE_1$ be the event that
all $n$ jobs are short.
The probability of a job $j$ being
long is $\pr(P_j > m) = 2^{-(m-1)} = 2/n^2$,
and thus, a union bound gives
\(
\pr(\cE_1) = 1 - n\cdot(2/n^2) = 1-2/n \ .
\)

Let $L := n^{3/4}$, and define $\cE_2$ to be the event that the
algorithm finishes at most $n-\frac{L}2$ jobs during the first
$t = \floor{2(n-L)}$ timesteps. Thus, under $\cE_2$ we have $\delta(t) \geq \frac{L}{2}$.

In the following, we condition on $\cE_1$ and assume that all jobs are short.
Note that every active short operation has an operation processing time of $1$,
hence an algorithm cannot distinguish between all jobs under $\cE_1$.
To analyze $\cE_2$ under $\cE_1$, consider the complementary event, and suppose
the algorithm finishes $k = n-\frac{L}2$ jobs in the first $t$
timesteps. Say these completed jobs are numbered $1, \ldots, k$,
this means that $\sum_{j = 1}^k P_j \leq t = \floor{2(n-L)}$. 
Rephrasing, this means that among the first $N := 2(n - L)$ unbiased coin flips
there are at least $n-\frac{L}2$ heads. 
However, the expected number of heads is only $\EX[X] = N/2 = n - L$,
hence we can show that this only happens with small probability.
Indeed, by Chebyshev's inequality, we have
\[
\pr\left[X \ge \EX[X] + \frac{L}{2}\right]
\le \frac{\VAR(X)}{(L/2)^2} 
\leq \frac{4n}{L^2} = o(1) \ ,
\]
because $\VAR(X)=N\cdot \frac12(1-\frac12)=N/4 \le n$ and $L=n^{3/4}$.
Thus, we have $\pr(\cE_2 \mid \cE_1) = 1-o(1)$.

Since also $\Pr(\cE_1) = 1-o(1)$, we conclude $\Pr(\cE_1 \cap \cE_2) = 1-o(1)$.
Then
  \[
  \EX[\delta(t)] \geq \frac{L}{2} \cdot \pr(\cE_1 \cap \cE_2) \geq
    \Omega(L) = \Omega(n^{3/4}) \ ,
\]
which proves the lemma.
\end{proof}

We can now prove~\Cref{theorem:randomized-lb}.

\begin{proof}[Proof of \Cref{theorem:randomized-lb}]
    \Cref{lemma:randomized-lb-alg,lemma:randomized-lb-opt} immediately imply 
    \begin{align*}
        \frac{\EX[\delta(t)]}{\EX[\delta^*(t)]} \ge \Omega(\log n) = \Omega(m).
    \end{align*}
    The theorem now follows by applying \Cref{prop:dos} to every realization of our distribution and finally using Yao's principle.
\end{proof}

\subsection{Stronger Lower Bound for Operations-SRPT}
\label{app:srpt:lb}

In this section, we show that the competitive ratio of Operations-SRPT is not constant, even if each job has two operations.

\thmSRPTlB*

The basic idea of our construction follows from the lower bound that is given in~\cite{BenderMR04} for an approximate version of SRPT.
This version of SRPT groups the jobs into classes based on their remaining processing time, i.e., job $j$ is in class $k$ at time $t$ if $p_j(t) \in [2^k,2^{k+1})$, and always processes some job from the smallest class. However, our lower bound has to take into account that Operations-SRPT is precise w.r.t.\ the remaining processing times of the active operations. This requires a more involved construction compared to~\cite{BenderMR04}. 

To construct the lower bound instance, we first observe that if there 
exists a \emph{critical} point in time $t$ at which Operations-SRPT and 
$\OPT$ satisfy certain properties, then there exists a sequence of job 
releases after time $t$ that strictly \emph{increases} the local competitive 
ratio of Operations-SRPT on the instance. We define these critical times as follows. 
For the remaining section, fix a sufficiently small constant $\varepsilon > 0$. We also assume w.l.o.g.\ that $\OPT$ is running SRPT.
Let $J(t)$ and $J^*(t)$ denote the sets of active jobs
at time $t$ in the schedule of Operations-SRPT and $\OPT$, respectively.

\begin{definition}
    \label{def:lb:critical}
    Fix parameters $M \in \mathbb{R}_+$ and $k \in \mathbb{N}_+$. We refer to a time $t$ as $(k,M)$-critical if the following three conditions hold:
        \begin{enumerate}
            \item $|J^*(t)| = 1$ and $p^*_j(t) =2M - \varepsilon$ for the single job $j \in J^*(t)$.
            \item There is a job $i \in J(t)$ such that $p_{i_1}(t) = 0$ and $p_{i_2}(t) = \frac{M}{2^k}-\varepsilon$.
            \item All active operations in $J(t)$ except for $i_2$ have remaining processing time at least $\frac{M}{2^k}$.
        \end{enumerate}
\end{definition}

The next lemma shows that if Operation-SRPT reaches a critical time in its schedule, then the lower bound instance can enforce another critical point in time with a worse local competitive ratio.

\begin{lemma}
    \label{lem:lb:recursion}
    Let $t$ be a $(k,M)$-critical time. Then, an adversary can enforce a $(k+1,M)$-critical time $t' = t + 2 \frac{M}{2^k} - \varepsilon$ with $|J(t')| = |J(t)|+1$ by releasing two jobs during $[t,t')$.
\end{lemma}

\begin{proof}
Let $t$ and $t'$ be as defined in the lemma, and consider the following job releases during $[t,t')$:
    \begin{enumerate}
        \item At time $t$, release a job $a$ with $p_{a_1} = \frac{M}{2^k} - \frac{\varepsilon}{2}$ and $p_{a_2} = \frac{M}{2^{k+1}} - \frac{\varepsilon}{2}$.
        \item At time $t_b = t + \frac{M}{2^k} + \frac{M}{2^{k+1}} -  \varepsilon$, release a job $b$ with $p_{b_1} = \frac{M}{2^{k+1}}$ and $p_{b_2} = 0$.
    \end{enumerate}
Next, we argue about the behavior of $\opt$ and Operations-SRPT between $t$ and $t'$:
\begin{itemize}
    \item First consider the time interval $[t,t_b]$: 
    Since $t$ is $(k,M)$-critical by assumption, there are exactly two jobs that $\opt$ can process at time $t$: The job $j$ with $p_j^*(t)= 2M-\varepsilon$ that exists because $t$ is $(k,M)$-critical by assumption and the job $a$ that is released at $t$. Since $a$ is the shorter job and no further jobs are released during $[t,t_b)$, $\opt$ will process $a$ and complete it by time $t_b$.

    Operations-SRPT will first work on the stage-$2$ operation of the job $i$ with $p_{i_2}(t) = \frac{M}{2^k}-\varepsilon$, which exists by the assumption that $t$ is $(k,M)$-critical. Since no further jobs are released, Operations-SRPT will run the job to completion. 
    
    The completion of $i_2$ takes $\frac{M}{2^k}-\varepsilon$ time units. For the remaining $\frac{M}{2^{k+1}}$ time units of the interval, Operations-SRPT reduces the remaining processing time of $a$'s stage-$1$ operation $a_1$ to $\frac{M}{2^{k+1}}- \frac{\varepsilon}{2}$.
    
    \item Next, consider the interval $[t_b,t']$. At the beginning of this interval, job $b$ is released.  Since $t'-t_b = \frac{M}{2^{k+1}}$, the optimal solution has enough time to complete $b$ during $[t_b,t']$. Hence, $J^*(t) = J^*(t')$ and $|J^*(t')| = 1$. Since the single job in $J^*(t) \cap J^*(t')$ is not processed during the interval, we still have $p_j^*(t') = 2M-\varepsilon$ by our assumption that $t$ is $(k,M)$-critical. Hence, $t'$ satisfies the first condition of~\Cref{def:lb:critical}.

    Starting from time $t_b$, Operations-SRPT continues to process the stage-$1$ operation of $a$, and completes it at time $t'-\frac{\varepsilon}{2}$. 
    After $a$ is completed, we have $p_{a_2}(t'-\frac{\varepsilon}{2}) = \frac{M}{2^{k+1}} - \frac{\varepsilon}{2}$,  and
    $p_{b_1}(t'-\frac{\varepsilon}{2}) = \frac{M}{2^{k+1}}$. 
    Hence, in the interval $[t' -\frac{\varepsilon}{2},t']$, Operations-SRPT works on the stage-$2$ operation of job $a$.
    At time $t'$, we have $p_{b_1}(t') = \frac{M}{2^{k+1}}$ and $p_{a_2}(t') = \frac{M}{2^{k+1}}-\varepsilon$. The existence of job $a$ implies  that $t'$ satisfies the second condition of~\Cref{def:lb:critical}.

    In order to show that $t'$ is $(k+1,M)$-critical, it only remains to show that the third condition of~\Cref{def:lb:critical} is satisfied, i.e., all jobs apart from $a$ have remaining time at least $\frac{M}{2^{k+1}}$ at time $t'$. For $b$, we already argued above that the remaining time is exactly $\frac{M}{2^{k+1}}$. Since we have $J(t') = (J(t)\setminus \{i\}) \cup \{a,b\}$, the assumption that $t$ is $(k,M)$-critical implies that all jobs in $J(t')\setminus \{a,b\}$ have remaining time at least $\frac{M}{2^k}$. Hence, the condition is satisfied.

    We conclude the proof of the lemma by observing that $J(t') = (J(t)\setminus \{i\}) \cup \{a,b\}$ implies $|J(t')| = |J(t)|+1$.\qedhere
\end{itemize}
\end{proof}

By repeatedly applying~\Cref{lem:lb:recursion}, an adversary can make the local competitive ratio of Operations-SRPT arbitrarily large, provided that Operations-SRPT reaches a first critical time at some point in its schedule. The next lemma shows that there indeed is an instance that forces Operations-SRPT to reach a critical time.

\begin{lemma}
\label{lb:instance:start}
Fix parameter $M \in \mathbb{R}_+$ and let $t = M$. There exists an instance that releases two jobs at time $0$ and forces time $t$ to be $(1,M)$-critical.
\end{lemma}

\begin{proof}
Consider the following two job releases at time $0$:
\begin{enumerate}
    \item Release a job $j$ with $p_{j_1} = p_{j_2} = M - \frac{\varepsilon}{2}$.
    \item Release a job $i$ with $p_{i_1} = M$ and $p_{i_2} = 0$.
\end{enumerate}
During $[0,M]$ the optimal solution will only work on job $i$ and complete it at time $M$. This implies $J^*(t) = \{j\}$ and $p^*_j(t)=2M-\varepsilon$. Hence, $t=M$ satisfies the first condition of~\Cref{def:lb:critical}.

Operations-SRPT, on the other hand, first completes the stage-$1$ operation of $j$ at time $M-\frac{\varepsilon}{2}$. Then, Operations-SRPT uses the remaining $\frac{\varepsilon}{2}$ time units to work on the stage-$2$ operation of job $j$ and reduces its remaining time to $M-\varepsilon$. The existence of $j$ implies that $t$ satisfies the second condition of~\Cref{def:lb:critical}. Finally, the fact that $p_{i_1}(t) = M$ implies that the third condition is satisfied as well. Hence, $t$ is $(1,M)$-critical.
\end{proof}

Having~\Cref{lb:instance:start,lem:lb:recursion} in place, we are ready to complete our lower bound construction and prove~\Cref{thm:srpt:lb}.

\begin{proof}[Proof of~\Cref{thm:srpt:lb}]

Fix a large $k^* \in \mathbb{N}$, define $M = 2^{k^*}$, and consider the following instance:
\begin{enumerate}
\item At time $0$, release two jobs as described in the proof of~\Cref{lb:instance:start}.
\item For each $k \in \{2,\ldots,k^*\}$:
\begin{enumerate}
    \item Consider time $t = M + \left(\sum_{k' = 2}^{k-1} 2 \cdot \frac{M}{2^{k'}}- \varepsilon\right)$ and $t'= t + 2 \frac{M}{2^k}-\varepsilon$.
    \item Release two jobs during $[t,t']$ as described in the proof of~\Cref{lem:lb:recursion}.
\end{enumerate}
\item Let $\hat{t}=  M + \left(\sum_{k' = 2}^{k^*} 2 \cdot \frac{M}{2^{k'}}- \varepsilon\right)$. At each time $t' \in \{\hat{t}+1-\varepsilon,\ldots,\hat{t}+1-\varepsilon+M\}$, release a job $j$ with $p_{j_1}=1$ and $p_{j_2}=0$. 
\end{enumerate}

We first argue that $|J(\hat{t})| = k^*+1$ and $|J^*(\hat{t})|=1$. Using~\Cref{lb:instance:start}, we get that $|J(M)| = 2$ and $|J^*(M)|=1$. For each time $t'$ considered in Step $2$ of the construction, we can apply~\Cref{lem:lb:recursion}. Hence, $|J^*(t')|=1$, whereas $|J(t')|$ increases by one for each such time $t'$. This gives us $|J(\hat{t})| = k^*+1$ and $|J^*(\hat{t})|=1$.

Furthermore, we can observe that $\hat{t}$ is $(k^*,M)$-critical. By~\Cref{def:lb:critical} and since $M = 2^{k^*}$, this implies that there is a single operation in $J(\hat{t})$ with a remaining time of $1-\varepsilon$, and other active operations have remaining time at least $1$. By time $\hat{t}' = \hat{t}+1-\varepsilon$, we have $|J(\hat{t}')| = k^*$, $|J^*(\hat{t}')|=1$, and all operations in $J(\hat{t}')$ have a remaining time at least one.
The third step of the construction implies $|J(t)| = k^*+1$ and $|J^*(t)|=2$ for each $t \in [\hat{t}',\hat{t}'+M]$.

Since $\opt$ never has more than two active jobs during $[0,\hat{t}'+2M-\epsilon]$ and completes the last job at time $\hat{t}'+2M-\epsilon$, we can conclude that $\opt \le 2 \cdot (\hat{t}+1-\varepsilon+M) + 2M -\epsilon \le 2 \cdot(2M+1-\varepsilon) + 2M \le 7M$, using that $\hat{t} \le 2M$. On the other hand, we can lower bound $\alg \ge k^* \cdot M$ by just considering the cost incurred during $[\hat{t}',\hat{t}+M]$. Hence, $\frac{\alg}{\opt} \in \Omega(k^*)$.

By observing that $n = 2 \cdot k^* + M = 2 \log_2(M) + M$, we can conclude with $\frac{\alg}{\opt} \in \Omega(\log(n))$.
\end{proof}

\subsection{Lower Bounds for Scheduling with (Non-Obligatory) Testing}
\label{sec:optional-tests}

In this section, we briefly discuss the scheduling model of \emph{scheduling with testing}~as introduced by Dürr et al~\cite{DurrEMM20}. In contrast to scheduling with obligatory tests~\cite{DogeasEL24}, which corresponds to our setting with $m=2$ operations, this model considers \emph{optional} tests. That is, the scheduler can decide whether it wants to test jobs or not. Formally, each job $j$ has a \emph{testing time} $c_j$, a \emph{worst-case processing time} $\bar{p}_j$, and an actual processing time $p_j \le \bar{p}_j$ which is initially unknown to the scheduler.  If the scheduler decides to test a job $j$, then it has to execute the test for $c_j$ time units. After the test is completed, the actual job processing time $p_j$ is revealed, and the scheduler can complete $j$ by processing it for $p_j$ time units. On the other hand, if the scheduler decides to not test $j$, then it has to complete $j$ by processing it for $\bar{p}_j$ time units. Hence, depending on whether a job is tested or not, it has to be processed for $c_j+p_j$ or $\bar{p}_j$ time units. Algorithms for scheduling with testing are usually analyzed via adversarial competitive analysis, and compared against an offline optimal solution that knows the actual processing times in advance.\footnote{We remark that there exists an alternative line of research on scheduling with stochastic tests~\cite{LeviMS19}. However, within the context of our work, the adversarial model is more relevant.}

Scheduling with testing has been studied for sum of completion time minimization in different machine environments~\cite{DurrEMM20,LiuLWZ23,GongCH24,AlbersE20,BuldS25,DameriusKLX023}, for minimizing the makespan~\cite{AlbersE21,GongGLM22,GongL21,GongFLSSZ25,DameriusKLX023}, and energy minimization~\cite{BampisDKLP21}. However, to the best of our knowledge, scheduling with testing has not yet been studied for flow time minimization. We show that this is for a good reason, by providing strong lower bounds for scheduling with testing to minimize the flow time on a single machine. 

For the special case of sum of completion time minimization on a single machine, the best possible (deterministic) competitive ratio is known to be between $1.8546$ and $2$ for unit test sizes~\cite{DurrEMM20}. For arbitrary test sizes, the best-known upper bound is $2.3166$~\cite{LiuLWZ23}. If we allow randomization, the lower bound becomes slightly weaker and the upper bounds slightly improve~\cite{DurrEMM20,LiuLWZ23}.

In contrast to these results for sum of completion time minimization, and also to our results for flow time minimization with obligatory tests (cf.~\Cref{sec:swt}), we show that no deterministic algorithm for flow time minimization with testing has a constant competitive ratio even for unit tests. 

The main reason for this strong lower bound is that, depending on the executed tests, a deterministic algorithm and the optimal solution can have a different total processing volume, which is in contrast to scheduling with obligatory tests where all algorithms have the same processing volume. In fact, the following lower bound holds even if we allow the algorithm to retrospectively rearrange its schedule into the optimal SRPT schedule for the tests it decided to execute. The proof of the following lower bound combines the lower bound instance of~\cite{MotwaniPT94} with ideas from the lower bounds in~\cite{DurrEMM20}.

\begin{lemma}
    Every deterministic algorithm for scheduling with testing to minimize the flow time on a single machine has a competitive ratio of $\Omega(n^{1/3})$, even for unit testing times.
\end{lemma}

\begin{proof}
    Consider the following adversarial strategy:
    \begin{enumerate}
        \item At time $t = 0$, release a total of $k$ jobs each with testing time $c_j = 1$ and worst-case processing time $\bar{p}_j = \alpha$. Select the actual processing times as follows:
        \begin{itemize}
            \item $k/2$ jobs have an actual processing time of $p_j = 0$. Use $S$ to refer to the set of these jobs.
            \item $k/2$ jobs have an actual processing time of $p_j = \alpha$. Use $L$ to refer to the set of these jobs.
        \end{itemize}
        \item At each time $t \in \{ (1+\alpha)\frac{k}{2}, (1+\alpha)\frac{k}{2} + 1, \ldots , (1+\alpha)\frac{k}{2} + k^2\}$, release a job $j$ with $c_j = 1$, $\bar{p}_j = 1$ and $p_j = 1$.
    \end{enumerate}

    Fix an arbitrary deterministic algorithm and let $0 \le x \le k$ denote the number of jobs in $S \cup L$ that the algorithm decides to test.  Let $X \subseteq S \cup L$ denote the set of jobs that the algorithm decides to test. 
    
    By definition of the model, the algorithm has to process each job $j \in S \cup L$ for 
    $$
        \hat{p}_j = \begin{cases}
            c_j + p_j & j \in X\\
            \bar{p}_j & j \not\in X
        \end{cases}
    $$
    time units. Consider the algorithm's schedule for these processing times $\hat{p}_j$. We assume that the algorithm schedules the jobs according to SRPT for the processing times $\hat{p}_j$. This is without loss of generality, as any deviation from SRPT can never decrease the objective value of the algorithm~\cite{Schrage68}.
    
    Since a deterministic algorithm cannot distinguish between the jobs in $S$ and $L$ before testing them, an adversary can force the first $\min\{x,k/2\}$ tests executed by the algorithm to all go to jobs in $L$. 
    We distinguish two cases:

    \begin{enumerate}
        \item     If $x \ge k/2$, then the algorithm tests all jobs in $L$. Hence, the algorithm processes each job $j \in L$ for $\hat{p}_j = 1+\alpha$ time units and each job in $S$ for $\hat{p}_j \ge 1$ time units. By time $t = \frac{k}{2} \cdot (1+\alpha)$, the algorithm can complete at most $ (1 +\frac{\alpha}{1+\alpha}) \frac{k}{2} $ jobs, the $\frac{k}{2}$ jobs in $S$ and $\frac{\alpha}{2(1+\alpha)} k$ jobs in $L$. Hence, at least $\frac{1}{2(\alpha+1)}k$ jobs in $S \cup L$ are still unfinished at $t = \frac{k}{2} \cdot (1+\alpha)$.
        \item     If $x \le k/2$, then the algorithm does not test any jobs in $S$. Hence, the algorithm processes each job $j \in S$ for $\hat{p}_j = \alpha$ time units and each job in $j \in L$ for $\hat{p}_{j} \ge \alpha$ time units. By time $t = \frac{k}{2} \cdot (1+\alpha)$, the algorithm can complete at most $\frac{1+\alpha}{2\alpha} k$ jobs. Hence, at least $\frac{\alpha-1}{2\alpha}k$ jobs in $S \cup L$ are still unfinished at $t = \frac{k}{2} \cdot (1+\alpha)$.
    \end{enumerate}

    Choosing $\alpha = \frac{1 + \sqrt{5}}{2}$, guarantees that the algorithm in both cases has at least $c \cdot k$ unfinished jobs at time $t = \frac{k}{2} \cdot (1+\alpha)$ with $c = \frac{3}{4}-\frac{\sqrt{5}}{4}$. Since we assume that the algorithms schedules according to SRPT, at most one of these jobs has been processed during $[0,t]$. Except of a single job, all unfinished jobs have a remaining processing time of at least one at $t$.  By the second step of the adversarial strategy, this means that, at any $t \in \{ (1+\alpha)\frac{k}{2}, (1+\alpha)\frac{k}{2} + 1, \ldots , (1+\alpha)\frac{k^2}{2}\}$, the algorithm has $\Omega(k)$ unfinished jobs. Integrating over time, this implies an objective value of $\Omega(k^3)$.

    Next, consider the optimal solution. The optimal solution tests exactly the jobs in $S$ and leaves all remaining jobs untested.  Thus, the jobs in $S \cup L$ have a total processing volume of $(1+\alpha) \cdot \frac{k}{2}$ and the optimal solution is able to complete all jobs in $S \cup L$ before any further jobs are released. 
    The total flow time for the jobs in $S \cup L$ is
    $$
    \sum_{i = 1}^{\frac{k}{2}} i \cdot \alpha + \sum_{i = \frac{k}{2}+1}^{k} i \le \alpha \cdot \sum_{i = 1}^{k} i  =  \frac{\alpha}{2} \cdot (k^2+k). 
    $$
    The $k^2$ remaining jobs outside $S \cup L$ all have a flow time of exactly $1$, as the optimal solution can complete each job before the next one is released. Hence, these jobs have a total flow time of $k^2$, which implies that the total flow time over all jobs is $\mathcal{O}(k^2)$.

    Combining both bounds, the competitive ratio of the algorithm is $\Omega(\frac{k^3}{k^2}) = \Omega(k)$. The total number of released jobs is $n = k + k^2 < k^3$. Thus, $k > n^{1/3}$, which completes the proof.
\end{proof}

\section{Missing Proofs of~\Cref{sec:swt}}
\label{app:swt}

We provide all proofs that are missing from~\Cref{sec:swt}.
Recall that, for a type $\gamma \in \{A,B\}$, $Q^*_{\gamma}(t)$ denotes the set of alive  \emph{type-$\gamma$ operations} at time $t$ in an optimal solution.
For a type $\gamma \in \{A,B\}$ and an operation stage $\ell \in \{1,2\}$, let $Q^*_{\ell \gamma}(t)$ denote the set of alive \emph{stage-$\ell$ type-$\gamma$ operations} at time $t$ in an optimal solution.
Define $Q^*_{\ell \gamma}(t,t') := Q^*_{\ell \gamma}(t) \cap Q^*_{\ell \gamma}(t')$, 
$\delta^*_{\ell \gamma}(t) := |Q^*_{\ell \gamma}(t)|$ and 
$\delta^*_{\ell \gamma}(t,t') := |Q^*_{\ell \gamma}(t,t')|$.

We start by proving basic properties of Operations-SRPT.

\begin{restatable}{observation}{obsTypeBInterrupt}
  \label{obs:typeB:interrupt}
  Once a type-B job $j$ is processed for the first time, it will never be interrupted.
\end{restatable}

\begin{proof}
Let $t$ denote the earliest time such that the algorithm processes the first operation $j_1$ of $j$ during $[t,t+1]$.
 At time $t$, all active operations $q \in J(t)$ have to satisfy $p_q(t) \ge \unit$, as otherwise $j_1$ with $p_{j_1}(t)=\unit$  would not be processed. Since $p_{j_1}(t+1) < \unit$ and $p_{j_2} < \unit$, no active operation in $J(t)$ will be processed before $j$ completes.

  This only leaves jobs $i$ with $r_i > t$ to potentially interrupt $j$. However, the stage-$1$ operations $i_1$ of such jobs $i$ all have a size of $\unit$, so they also will not interrupt the processing of $j$.
\end{proof}

\obsSinglePartial*

\begin{proof}
  Following the same argumentation as in the previous observation, we can argue that a job with a remaining processing time less than $\unit$ will never be interrupted. This implies that there can be at most one such job at any point in time.
\end{proof}

We continue by showing that the volume invariant~\eqref{eq:volume:eq} indeed implies local competitiveness at time $\ts$.

\lemVolumeConversion*

\begin{proof}
    First, consider the following more fine-grained characterization of $\vol_\ts^*$, which uses the assumption of the lemma:
    \begin{align*}
      \vol_\ts^* 
      &=  \vol^*_\ts(Q^*_{2A}(\ts)) + \vol^*_{\ts}(Q^*_{B}(\ts)) + \vol^*_{\ts}(Q^*_{1A} (\ts)) \\
      &\le \vol_\tau(L_{2A}(|Q^*_{2A}(\ts)|,\ts)) + \vol^*_{\ts}(Q^*_{B}(\ts)) + \vol^*_{\ts}(Q^*_{1A} (\ts)) \ .
    \end{align*}  
    Rearranging yields
    \begin{equation}\label{eq:2}
      \vol_\ts^*-  \vol_\tau(L_{2A}(|Q^*_{2A}(\ts)|,\ts))  \leq  \vol^*_{\ts}(Q^*_{B}(\ts)) + \vol^*_{\ts}(Q^*_{1A})(\ts) \ .
    \end{equation}
    We continue to upper bound the different terms on the right-hand side of Inequality~\eqref{eq:2}:
    \begin{enumerate}
      \item  $\vol^*_{\ts}(Q^*_{1A} (\ts)) \le \unit |Q^*_{1A} (\ts)| \le \unit |Q^*_{2A} (\ts)| $: The first inequality is 
      because all stage-$1$ operations have remaining processing times of at most $\unit$. The second inequality is because each active stage-$1$ operation at $\ts$ must have a corresponding stage-$2$ operation in $Q^*_{2A} (\ts)$.
      \item  $\vol^*_{\ts}(Q^*_{B}(\ts)) < 2\unit \cdot |Q^*_{2B}(\ts)|$. The number of alive type-B jobs at time $\ts$ in the optimal solution is $|Q^*_{2B}(\ts)|$. Since all such jobs have a remaining processing time of strictly less than $2\unit$, this gives us a total volume of strictly less than $2p \cdot |Q^*_{2B}(\ts)|$.
    \end{enumerate}
    We now plug these upper bounds into~\eqref{eq:2} and obtain:
    \begin{align*}
    \vol_\ts^*-  \vol_\tau(L_{2A}(|Q^*_{2A}(\ts)|,\ts))
      \le  \unit|Q^*_{2A} (\ts)| + 2\unit \cdot |Q^*_{2B}(\ts)| \ .
    \end{align*}

    Since $\vol_\ts = \vol_\ts^*$, this implies $\vol_\ts  \le \unit|Q^*_{2A} (\ts)| + 2\unit \cdot |Q^*_{2B}(\ts)| + \vol_\tau(L_{2A}(|Q^*_{2A}(\ts)|,\ts))$. 
    By assumption of the lemma and using that the volume of a job is at least the volume of its second operation, we get that the jobs that correspond to the operations in $L_{2A}(|Q^*_{2A}(\ts)|,\ts)$ have a total remaining volume of at least $\vol_\tau(L_{2A}(|Q^*_{2A}(\ts)|,\ts))$. 
    Hence, all other jobs must have a remaining volume of at most $\vol_\ts -  \vol_\tau(L_{2A}(|Q^*_{2A}(\ts)|,\ts)) \le \unit|Q^*_{2A} (\ts)| + 2\unit \cdot |Q^*_{2B}(\ts)|$.
    As the algorithm by~\Cref{obs:single:partial} has at most one job with a remaining processing time less than $\unit$ at time $\ts$, this volume is only enough for at most $|Q^*_{2A} (\ts)| + 2 \cdot |Q^*_{2B}(\ts)|$ jobs. Together with the $|Q^*_{2A} (\ts)|$ jobs that correspond to the stage-$2$ operations in $L_{2A}(|Q^*_{2A}(\ts)|,\ts)$, this gives us a total number of jobs
    $$
    |J(\ts)| \le 2 \cdot |Q^*_{2A} (\ts)|  + 2 \cdot |Q^*_{2B}(\ts)|  = 2 \cdot |Q^*_2(\ts)| = 2 \cdot |J^*_2(\ts)| \ ,
    $$
    and completes the proof.
\end{proof}

Given~\Cref{lem:volume:conversion}, it suffices to show that the volume invariant~\eqref{eq:volume:eq} holds in order to prove~\Cref{thm:unit:locally:2} and, thus,~\Cref{thm:main-swt}. That is, it suffices to show $\vol_\tau (L_{2A}(|Q^*_{2A}(\ts)|,\ts)) \ge \vol^*_\ts(Q^*_{2A}(\ts))$.
As outlined in~\Cref{sec:swt}, a common approach (see e.g.~\cite{Schrage68,BansalD07,AzarLT21}) for proving such an invariant is to show that
$$
\vol_t(L_{2A}(|Q^*_{2A}(t) \cap Q^*_{2A}(\ts)|,t)) \ge \vol^*_t(Q^*_{2A}(t) \cap Q^*_{2A}(\ts))
$$
holds at any time $0 \le t \le \ts$ via induction, which then directly implies~\eqref{eq:volume:eq} for $t = \ts$. Note that this corresponds to inequality~\eqref{eq:volume:inductive} in~\Cref{sec:swt}. 

A critical difference between our setting and previous works employing this proof strategy is that there can be a time $t \in [0,\ts]$ such that the algorithm processes an operation $q \in L_{2A}(|Q^*_{2A}(t) \cap Q^*_{2A}(\ts)|,t)$ during $[t,t+1]$, but we still have $\vol_t(L_{2A}(|Q^*_{2A}(t) \cap Q^*_{2A}(\ts)|,t)) < \vol_t$. 
However, by exploiting the tie breaking rule, we can observe that this can only happen if $p_q(t) < \unit$ 

\begin{restatable}{observation}{obsT}
    \label{obs:t_0}
    Let $t \le \ts$ be a time such that Operations-SRPT processes an operation $q \in L_{2A}(|Q^*_{2A}(t) \cap Q^*_{2A}(\ts)|,t)$ with $p_q(t) \ge \unit$ during $[t,t+1]$. Then, $\vol_t(L_{2A}(|Q^*_{2A}(t) \cap Q^*_{2A}(\ts)|,t)) = \vol_t \ge \vol^*_t(Q^*_{2A}(t) \cap Q^*_{2A}(\ts))$.
\end{restatable}

\begin{proof}
 By the definition of Operations-SRPT, 
 an operation $q \in L_{2A}(|Q^*_{2A}(t) \cap Q^*_{2A}(\ts)|,t)$ with $p_q(t) \ge \unit$ is only processed if no other operations in addition to $L_{2A}(|Q^*_{2A}(t) \cap Q^*_{2A}(\ts)|,t)$ are alive at time $t$. 
 Note that this uses the tie breaking rule that prefers stage-$1$ operations over stage-$2$ operations. 
 This immediately implies $\vol_t(L_{2A}(|Q^*_{2A}(t) \cap Q^*_{2A}(\ts)|,t)) = \vol_t = \vol_t^* \ge \vol^*_t(Q^*_{2A}(t) \cap Q^*_{2A}(\ts))$.
\end{proof}

To address the additional challenge described above, we will consider a carefully chosen point in time $t_0 \le \ts$ and adjust our proof strategy depending on whether there exists a time $t \in [t_0,\ts]$ at which Operations-SRPT processes an operation $q \in L_{2A}(|Q^*_{2A}(t) \cap Q^*_{2A}(\ts)|,t)$ with $p_q(t) < \unit$. 

\begin{definition}[Critical time $t_0$]
    \label{def:unit_t0}
    Let $t_0$ be the latest point in time $t_0 \le \ts$ such that Operations-SRPT processes an operation $q \in L_{2A}(|Q^*_{2A}(t_0) \cap Q^*_{2A}(\ts)|,t_0)$ with $p_q(t_0) \ge \unit$ during $[t_0,t_0+1]$. If such a point in time does not exist, then define $t_0 := 0$. 
\end{definition}

If $t_0 \not= 0$, then $t_0$ is defined as the latest time $\le \ts$ at which we can apply~\Cref{obs:t_0}. Hence, Inequality~\eqref{eq:volume:inductive} holds for $t = t_0$, either by~\Cref{obs:t_0} or trivially if $t_0 = 0$.  
In case that there is no $t \in [t_0,\ts]$ such that the algorithm processes an operation $q \in L_{2A}(|Q^*_{2A}(t) \cap Q^*_{2A}(\ts)|,t)$ with $p_q(t) < \unit$ during $[t,t+1]$, 
we show that~\eqref{eq:volume:inductive} holds for every $t \in [t_0,\ts]$ by essentially replicating the inductive proof of Schrage's SRPT analysis~\cite{Schrage68}, 
but starting the induction at time $t_0$ instead of time $0$. 
This gives the following lemma, which we prove in~\Cref{app:swt:2}.

\begin{restatable}{lemma}{lemUnitA}
 \label{lem:unit:A:volume}
     If there is no time $t \in [t_0,\ts]$ such that the algorithm processes an operation $q \in L_{2A}(|Q^*_{2A}(t) \cap Q^*_{2A}(\ts)|,t)$ with $p_q(t) < \unit$ during $[t,t+1]$, then 
     $\vol_\ts(L_{2A}(|Q^*_{2A}(\ts)|,\ts)) \ge \vol^*_{\ts}(Q^*_{2A}(\ts))$.
\end{restatable}

It remains to consider the case where there is a time $t \in [t_0,\ts]$ such that the algorithm processes an 
operation $q \in L_{2A}(|Q^*_{2A}(t) \cap Q^*_{2A}(\ts)|,t)$ with $p_q(t) < \unit$ during $[t,t+1]$. 
For this case, we show that $|J^*(\ts)| = |J(\ts)|$. 
To do so, we first observe that~\Cref{obs:t_0} implies $|J(t_0)| = |J^*(t_0)|$ and $J^*(t_0) \subseteq J^*(\ts)$ (cf.~\Cref{lem:unit:t0:n-jobs} in~\Cref{app:swt:1}). That is, at $t_0$ both the algorithm and the optimal solution have the same number of alive jobs, and the optimal solution completes jobs in $J^*(t_0)$ only after time $\ts$. In~\Cref{app:swt:1}, we prove the stronger statement that neither the algorithm nor the optimal solution complete any jobs during $[t_0,\ts]$, which implies the following lemma.

\begin{restatable}{lemma}{lemTOptimal}
\label{lemma:t0:optimal}
    If there is a time $t \in [t_0,\ts]$ such that the algorithm processes an operation $q \in L_{2A}(|Q^*_{2A}(t) \cap Q^*_{2A}(\ts)|,t)$ with $p_q(t) < \unit$ during $[t,t+1]$, then $|J^*(\ts)| = |J(\ts)|$.
\end{restatable}

The~\Cref{lemma:t0:optimal,lem:unit:A:volume,lem:volume:conversion} imply \Cref{thm:unit:locally:2} and thus \Cref{thm:main-swt}. For the remainder of this section, we prove the~\Cref{lemma:t0:optimal,lem:unit:A:volume,lem:volume:conversion}.

To this end, we first prove the following auxiliary lemma regarding the point in time $t_0$ as defined in~\Cref{def:unit_t0}. As outlined in~\Cref{sec:swt}, we use a different proof strategy depending on whether there is a time $t \in [t_0,\ts]$ at which the algorithm processes an operation $q \in L_{2A}(|Q^*_{2A}(t) \cap Q^*_{2A}(\ts)|,t)$ with $p_q(t) < \unit$. The lemma proves that such a time $t$ can only exist if $t_0 \not= 0$ and $p_{q_0}(t_0) > \unit$ for the operation $q_0$ that is processed at time $t_0$.

\begin{lemma}
    \label{lem:unit:no:below:bar}
    If $t_0 = 0$, or $t_0 > 0$ and $p_{q_0}(t_0) > \unit$ for the operation $q_0$ that is processed by the algorithm at time $t_0$, then each $t > t_0$ satisfies that no operation in $L_{2A}(\delta^*_{2A}(\ts,t),t)$ is processed by the algorithm at $t$.
\end{lemma}

\begin{proof}
    For every $t > t_0$, \Cref{def:unit_t0} immediately implies that no operation $q \in L_{2A}(\delta^*_{2A}(\ts,t),t)$ with $p_{q}(t) \ge \unit$ is processed at time $t$. It remains to show that the same holds for operations $q \in L_{2A}(\delta^*_{2A}(\ts,t),t)$ with $p_q(t) < \unit$.

    For the sake of contradiction, assume that there is a time $t > t_0$ such that an operation $q \in L_{2A}(\delta^*_{2A}(\ts,t),t)$ with $p_q(t) < \unit$ is processed at $t$. 
    Assume that $t > t_0$ is the earliest such time.
    By definition of Operations-SRPT and since all the stage-$1$
    operations have a processing time of exactly $\unit$ upon arrival, the
    choice of $t$ implies that $q$ is processed during $[t-1,t]$ and $p_q(t-1) = \unit$.  If $q \in L_{2A}(\delta^*_{2A}(\ts,t-1),t-1)$, then we immediately arrive at a contraction, either to the definition of $t_0$ or to our assumption that $p_{q_0}(t_0)> \unit$. 

    Thus, we must have $q \not\in L_{2A}(\delta^*_{2A}(\ts,t-1),t-1)$ but $q \in L_{2A}(\delta^*_{2A}(\ts,t),t)$. Since $q$ is the only operation that is processed during $[t-1,t]$, the only way that this can happen is if some job $j$ is released at $t-1$. However, this would imply that the stage-$1$ operation $j_1$ is active at $t-1$.
    By the tie breaking rule of Operations-SRPT, this is a contradiction to the second type-A operation $q$ with $p_q(t-1) = 1$ being processed during $[t-1,t]$.
\end{proof}

\subsection{Proof of~\Cref{lemma:t0:optimal}}
\label{app:swt:1}

We continue by proving~\Cref{lemma:t0:optimal}, which we restate here for the sake of convenience.

\lemTOptimal*

In the following, we use $q_0$ to refer to the operation which is processed at time $t_0$. By~\Cref{lem:unit:no:below:bar} and the assumption of~\Cref{lemma:t0:optimal}, we have $p_{q_0}(t_0) = \unit$ and $t_0 \not= 0$. This allows us to use the following  two auxiliary lemmas, which we will use to prove~\Cref{lemma:t0:optimal}.

\begin{lemma}
\label{lem:unit:t0:n-jobs}
 If $t_0 \not= 0$ and $p_{q_0}(t_0) \ge \unit$ for the operation $q_0$ that is processed by the algorithm at time $t_{q_0}$, then $|J(t_0)| = |J^*(t_0)| = |J^*(t_0) \cap J^*(\ts)|$.
\end{lemma}

\begin{proof}
By~\Cref{obs:t_0}, we know that 
$\vol^*_{t_0} = \vol_{t_0} = \vol_{t_0}(L_{2A}(\delta^*_{2A}(\ts,t_0),t_0))$,
which implies $|J(t_0)| \le \delta^*_{2A}(\ts,t_0) = |Q^*_{2A}(t_0) \cap Q^*_{2A}(\ts)| \le |J^*(t_0) \cap J^*(\ts)| \le |J^*(t_0)|$. 
Furthermore, we know that the optimal solution executes SRPT, and thus, is locally $1$-competitive. This implies $|J(t_0)| \ge |J^*(t_0)|$ and hence $|J(t_0)| = |J^*(t) \cap J^*(t_0)| = |J^*(t_0)|$.
\end{proof}

\begin{lemma}
    \label{lem:unit:small-A-execution}
    If $t_0 \not= 0$ and $p_{q_0}(t_0) = \unit$ for the operation $q_0$ that is processed by the algorithm at time $t_0$, then $\ts-t_0 < p_{q_0}(t_0)$, i.e., the algorithm cannot finish $q_0$ by $\ts$.
\end{lemma}

\begin{proof}
For the sake of contraction, assume $\ts-t_0 \ge p_{q_0}(t_0)$ and let $t = t_0 + p_{q_0}(t_0) = t_0 + \unit$. We can observe the following facts:
\begin{enumerate}
\item The algorithm will complete operation $q_0$ at point in time $t$: Since $q_0$ with $p_{q_0}(t_0)= \unit$ is processed during $[t_0,t_0+1]$ by assumption, we have $p_{q_0}(t_0+1) < \unit$. By the algorithm definition and since the stage-$1$ operations of newly arriving jobs have a processing time equal to $\unit$, $p_{q_0}(t_0+1) < \unit$ implies that the operation will be processed until it is completed. Since $q_0$ is a stage-$2$ operation, the corresponding job also completes at time $t$.
\item The optimal solution will not complete any job during $[t_0,t]$: By~\Cref{lem:unit:t0:n-jobs}, we have $J^*(t_0) \subseteq J^*(\ts)$, which implies that all jobs in $J^*(t_0)$ are still alive at time $\ts$ and thus also at time $t$. This only leaves jobs $j$ with release dates $r_j > t_0$ to be potentially completed by the optimal solution during $[t_0,t]$. However, $r_j > t_0$ implies $t - r_j < \unit$. Since all jobs have a processing time of at least $\unit$, the interval $[r_j,t]$ is not long enough to finish any job with $r_j > t_0$.
\end{enumerate}

Let $d$ denote the number of jobs that are released during $(t_0,t]$. The two above facts imply $|J(t)| \le |J(t_0)| - 1 + d = |J^*(t_0)| - 1 + d$ and $|J^*(t)| = |J^*(t_0)| + d$. However, then $|J(t)| < |J^*(t)|$, which is a contradiction to the optimal solution (SRPT) being locally $1$-competitive.
\end{proof}

\begin{proof}[Proof of~\Cref{lemma:t0:optimal}]
By \Cref{lem:unit:t0:n-jobs}, we know $|J(t_0)| = |J^*(t_0)|$. Furthermore, by~\Cref{lem:unit:small-A-execution}, we have $\ts-t_0 < p_{q_0}(t_0) = \unit$, which implies that the algorithm does not complete any operations during $[t_0,\ts]$. Following the arguments in the proof of~\Cref{lem:unit:small-A-execution}, we can argue that the optimal solution also does not complete any jobs during  $[t_0,\ts]$. Thus, $|J(\ts)| = |J^*(\ts)|$.
\end{proof}

\subsection{Proof of~\Cref{lem:unit:A:volume}}
\label{app:swt:2}

In this section, we prove~\Cref{lem:unit:A:volume}, which we restate here for the sake of convenience.

\lemUnitA*

In the following, we assume w.l.o.g.~that $t_0 \not= \ts$. Otherwise,~\Cref{lem:unit:A:volume} follows immediately from~\Cref{obs:t_0}.

Next, we observe that the assumption of the lemma implies that either $t_0 = 0$ or $p_{q_0}(t_0) > \unit$. 

\begin{observation}
    \label{obs:unit:no:below:bar}
    If there is no time $t \in [t_0,\ts]$ such that the algorithm processes an operation $q \in L_{2A}(|Q^*_{2A}(t) \cap Q^*_{2A}(\ts)|,t)$ with $p_q(t) < \unit$ during $[t,t+1]$, then either $t_0 = 0$ or $p_{q_0}(t_0) > \unit$.
\end{observation}

\begin{proof}
For the sake of contradiction, assume $t_0 \not=0$ and $p_{q_0}(t_0) = \unit$ for the operation $q_0$ that is processed at time $t_0$. Since $t_0 \not= 0$, the operation $q_0$ is processed during $[t_0,t_0+1]$ by~\Cref{def:unit_t0}. This implies $p_{q_0}(t_0+1) < \unit$. However, this implies that $q_0$ is also processed during $[t_0+1,t_0+2]$; a contradiction to the assumption that there is no time $t \in [t_0,\ts]$ such that the algorithm processes an operation $q \in L_{2A}(|Q^*_{2A}(t) \cap Q^*_{2A}(\ts)|,t)$ with $p_q(t) < \unit$ during $[t,t+1]$.
\end{proof}

The~\Cref{obs:unit:no:below:bar} will allow us to use~\Cref{lem:unit:no:below:bar} in the proof of~\Cref{lem:unit:A:volume}.

Finally, we are ready to prove~\Cref{lem:unit:A:volume} by replicating the inductive proof of the classical SRPT analysis by~\cite{Schrage68}. The only difference is that we start our induction at time $t_0$ instead of time $0$.

\begin{proof}[Proof of~\Cref{lem:unit:A:volume}]
    We show that $\vol_t(L_{2A}(\delta^*_{2A}(\ts,t),t)) \ge \vol^*_{t}(Q^*_{2A}(t))$ holds for every $t_0 \le t \le \ts$, which then implies the lemma.

    \textbf{Base case:} Assume $t = t_0$. If $t_0 = 0$, then the statement holds trivially. Otherwise, i.e., if $t_0 > 0$, then we have $\vol_t^* = \vol_t = \vol_t(L_{2A}(\delta^*_{2A}(\ts,t),t))$ by~\Cref{obs:t_0} and can conclude with $\vol_t(L_{2A}(\delta^*_{2A}(\ts,t),t)) = \vol_t^* \ge \vol^*_{t}(Q^*_{2A}(t))$.

    \textbf{Induction step:} Consider a $t > t_0$. By induction hypothesis, we have $\vol_{t-1}(L_{2A}(\delta^*_{2A}(\ts,t-1),t-1)) \geq \vol^*_{t-1}(Q^*_{2A}(\ts,t-1))$.

    Furthermore, we know by~\Cref{lem:unit:no:below:bar} that the operations in $L_{2A}(\delta^*_{2A}(\ts,t-1),t-1)$ are not processed during $[t-1,t]$. Hence, the volume of these operations does not change between time $t-1$ and time $t$, and we get
    $$
    \vol^*_{t-1}(Q^*_{2A}(\ts,t-1)) \le \vol_{t-1}(L_{2A}(\delta^*_{2A}(\ts,t-1),t-1)) = \vol_t(L_{2A}(\delta^*_{2A}(\ts,t-1),t-1)).
    $$
    If no type-A jobs are released at time $t$, then we have $L_{2A}(\delta^*_{2A}(\ts,t-1),t-1) = L_{2A}(\delta^*_{2A}(\ts,t),t)$ and $Q^*_{2A}(\ts,t-1) = Q^*_{2A}(\ts,t)$. Since the remaining volume of $Q^*_{2A}(\ts,t-1)$ can never increase from time $t-1$ to time $t$, we can conclude with 
    \begin{align*}
    \vol_{t}(L_{2A}(\delta^*_{2A}(\ts,t),t))
    = & \ \vol_{t-1}(L_{2A}(\delta^*_{2A}(\ts,t-1),t-1))\\ 
    \ge & \ \vol^*_{t-1}(Q^*_{2A}(\ts,t-1))\\
    \ge & \ \vol^*_{t}(Q^*_{2A}(\ts,t)).
    \end{align*}
    It remains to consider the case where type-A jobs are released at time $t$. Let $D$ denote the set of all stage-$2$ type-A operations that belong to jobs released at time $t$, and partition $D$ into $D_1 := D\cap Q^*_{2A}(\ts)$ and $D_2 := D \setminus D_1$. By definition, and again using that the remaining volume of $Q^*_{2A}(\ts,t-1)$ can never increase from time $t-1$ to time $t$, we get
    \begin{equation}
        \label{eq:ind:1}
        \vol^*_{t}(Q^*_{2A}(\tau,t)) \le \vol^*_{t-1}(Q^*_{2A}(\tau,t-1)) + \sum_{q \in D_1} p_q
    \end{equation}
    and $|Q^*_{2A}(\tau,t)| = |Q^*_{2A}(\tau,t-1)| + |D_1|$. The latter immediately implies $|L_{2A}(\delta^*_{2A}(\ts,t),t)| =  |L_{2A}(\delta^*_{2A}(\ts,t-1),t-1)| + |D_1|$. Since (i) we already argued that the remaining volume of $L_{2A}(\delta^*_{2A}(\ts,t-1),t-1)$ does not decrease from $t-1$ to $t$ and (ii)  $L_{2A}(\delta^*_{2A}(\ts,t),t)$ is defined to contain the $|Q^*_{2A}(\tau,t)|$-\emph{largest} alive stage-$2$ type-A operations, the increased cardinality of $L_{2A}(\delta^*_{2A}(\ts,t),t)$ compared to $L_{2A}(\delta^*_{2A}(\ts,t-1),t-1)$ implies the following volume increase: 
    \begin{equation}
        \label{eq:ind:2}
        \vol_{t}(L_{2A}(\delta^*_{2A}(\ts,t),t)) \ge \vol_{t-1}(L_{2A}(\delta^*_{2A}(\ts,t-1),t-1)) + \sum_{q \in D_1} p_q.
    \end{equation}
    Combining~\eqref{eq:ind:1} and~\eqref{eq:ind:2} and plugging in the induction hypothesis yields the required 
    $$\vol_t(L_{2A}(\delta^*_{2A}(\ts,t),t)) \ge \vol^*_{t}(Q^*_{2A}(t)).$$
    
    \end{proof}

\section{Missing Proofs of~\Cref{sec:general}}

\subsection{Properties of the Algorithm}
\label{app:properties}

We first state some basic properties of the algorithm. We remark that all these properties are analogous to properties of the algorithm in~\cite{GuptaKPW26} and can be shown by following the corresponding proofs given in~\cite{GuptaKPW26} but arguing about chunks instead of jobs, e.g., replace the class of a job with the class of the currently active chunk of the job. For the sake of completeness, we still give the adjusted proofs in~\Cref{app:properties}. The first two properties follow directly from the definition of the algorithm.

\begin{restatable}{fact}{factMonotoneStack}\label{fact:monotone-stack}
	At any time, 
	let $c_1,\ldots,c_{z}$ be the chunks in the stack $\apart$ at time $t$ indexed by their position in the stack, i.e., $\topjob = c_1$.
	Then, $\curclass_{c_1}(t) < \ldots < \curclass_{c_{z}}(t)$.
\end{restatable}

\begin{proof}
	This is because a chunk $c$ can only move to $\apart$ at time $t$ if its current class $\curclass_{c}(t)$ is strictly smaller than the current class of $\topjob$.
\end{proof}

\begin{restatable}{fact}{factMoveSmallest}\label{fact:move-smallest}
	If a chunk moves from $\afull$ to the stack $\apart$ at time $t$, it has the smallest current class of all active chunks at time $t$.
\end{restatable}

\begin{proof}
	By \Cref{fact:monotone-stack}, the chunk has the smallest class of all chunks in $\apart$ after being moved to $\apart$. Furthermore, since $\afull$ is sorted by current class, the moved chunk also has the smallest class of all chunks in $\afull$ before being moved. 
\end{proof}

The next lemma formulates properties of the chunk $c$ that is processed at a time $t$. In particular, such a chunk has the smallest class among all active chunks with possibly a single exception.

\lemStrictClasses*

\begin{proof}
	Note that the second statement of the lemma directly follows from the first statement.
	We now prove the first statement.
	By \Cref{fact:monotone-stack}, every chunk in $J(t)$
	of class $\leq k(t)$ must be full.
	Thus, let $c_1 \in \afull(t)$ be a chunk of class $k_1 < k(t)$.
	Since $c$ is processed during $[t,t+1]$, the condition of the algorithm for moving chunks to $\apart(t)$ must have been wrong, as otherwise another chunk would have been moved to $\apart(t)$ before processing $c$.
	Since $c_1$ is in $\afull(t)$ and has class $< k(t)$, the only way for the condition to be wrong is that $|\afull(t)| < |J(t)|/4$. 
	
	Since $|\afull(t)| \geq 1$, this implies that $|\apart(t)| \geq 2$. 
	Thus, let $c'$ be the chunk directly under $c$ in $\apart(t)$ and let $t'$ be the time when $c'$ was moved to the stack. Note that $t' \leq t$.
	Let $C'$ be the chunks in $J(t)$ that became active after $t'$.
	Note that if a chunk in $C' \setminus \{c\}$ is in $\apart(t)$, then it must be below $c'$ in $\apart(t)$, which is impossible as $c'$ is already in the stack when such a chunk becomes active by definition of $t'$.
	Thus, $C' \setminus \{c\} \subseteq \afull(t)$.
	Since $c'$ was moved to the stack at time~$t'$, the condition of the algorithm was satisfied before the move. In particular,
	\[
	    |\afull(t')| \geq \frac14 \cdot |J(t)'| \ .
	\]
	We now compute how this condition changes until time $t$.
	There are $|C'| - 1$ chunks of $\afull(t)$ becoming active 
	during $[t',t]$ (the $-1$ is because $c \notin \afull(t)$).
	Moreover, $c'$ is moved from $\afull(t')$ to $\apart(t')$ at time $t'$. Thus, the LHS increases by $|C'| - 2$.
	On the other hand, the RHS increases by $\frac14 \cdot |C'|$ due to chunks in $J(t)$ becoming active during $[t',t]$.
	
	For the sake of contradiction, assume that there exists another chunk $c_2 \neq c_1$ of class $k_2 \leq k(t)$ in $\afull(t)$.
    We have $k_1, k_2 \le k(t) < k_{c'}$. Since $c'$ is in $\apart(t)$, all three chunks $c$, $c_1$, and $c_2$ must have become active \emph{after} $t'$, and thus, be part of $C'$. Thus, $|C'| \geq 3$ and
	\[
	    |\afull(t)| = |\afull(t')|  + |C'| - 2
					\geq \frac14 \cdot |J(t)'| + \frac14 \cdot |C'| = \frac14 \cdot |J(t)| \ .
	\]
	But this contradicts our earlier observation that $|\afull(t)| < |J(t)|/4$.
\end{proof}

\begin{restatable}{lemma}{lemBalance}\label{lem:alg:property}
	At any time $t$, it holds that $|\afull(t)| \geq \frac{|J(t)|}{4} - 1$. 
\end{restatable}

\begin{proof}
	The proof is by induction on $t$. It is trivially true at $t=0$.
	Suppose it holds for time $t$. We show it also holds for time $t+1$ by considering all possible events that can happen between time $t$ and $t+1$:
	\begin{itemize}
		\item
		If a job arrives at time $t$, then $|\afull(t)|$ increases by $1$, but $|J(t)| / 4$ increases by $1/4$. Thus, the inequality still holds.
		\item
		If we move a chunk from $\afull(t)$ to $\apart(t)$ at time $t$, then it must be that $|\afull(t)| \geq \frac{|J(t)|}{4}$
		before the move. Thus, the inequality still holds after the move, as $|J(t)|$ does not change.
		\item 
		If a job finishes due to being processed during $[t,t+1]$, then $|J(t+1)| = |J(t)| - 1$, but $|\afull(t)|$ does not change. Thus, the inequality still holds after the job finishes.
		\item If a chunk of a job finishes during $[t,t+1]$, but the job does not finish, then we move the job (more precisely its next chunk) to $\afull(t+1)$. Since $|J(t)|$ does not change,
		the inequality still holds after the move.
\end{itemize}	
This covers all cases, and thus, the proof.
\end{proof}

\subsection{Missing Proofs of~\Cref{sec:chunks}}
\label{app:chunks}

\obsAlgChunks*

\begin{proof}
    Fix a job $j$. We first show that $j$ is inserted into $\afull$ at time $t$ if and only if $t \in \{r_{c_1},\ldots,r_{c_\ell}\}$. 
    Clearly, the first time $j$ is added to $\afull$ by the algorithm is its release date $r_j = r_{c_1}$. 
    By Step~3.~of the algorithm, $j$ is re-added to $\afull$ at time $t$ if and only if an operation $j_i$ of $j$ completes at time $t$ and $k_{j_{i+1}} > \max_{i' \le i} k_{j_{i'}}$ holds for the operation $j_{i+1}$ that becomes active at $t$.  By~\Cref{def:chunk}, the latter holds if and only if $j_{i+1}$ is the first operation of a chunk $c$ of $j$, which is the case if and only if $t = r_c$.

    Next, we prove the additional properties of the observation:
    \begin{itemize}        
        \item By Step~1.~of the algorithm, job $j$ can only be added to $\apart$ if it is part of $\afull$. Similarly, with the exception of $j$'s release date, $j$ can only be added to $\afull$ if it is in $\apart$. Since $j$ can also only complete if it is in $\apart$, the first property follows.
        
        \item Upon arrival of $j$, we have $\curclass_j = k_{j_1} = k_{c_1}$ by definition of the algorithm. Afterwards, the class of $j$ is updated to the class of the currently active operation of $j$ if and only if $j$ is re-inserted into $\afull$. As we already showed, the re-insertion of $j$ coincides with the time some chunk $c$ of $j$ becomes active. Hence, the class of $j$ will be updated to the class of $c$'s first operation, which is equivalent to the class of $c$.
    \end{itemize}
\end{proof}

\subsection{Missing Proofs of~\Cref{sec:local}}
\label{app:chunk-lp}

\factLPRelax*

\begin{proof}
    Consider an optimal solution for the given instance of the \ofts. To prove the statement, we construct a feasible solution of (\hyperlink{localchunklp}{$\chunklp(\ts)$}) with objective value $|J^*_c(\tau)|$.
    To this end, assume that each $c \in J_c^*(\tau)$ satisfies $p^*_c(\tau) = p_c$, i.e., the chunks in $J_c^*(\tau)$ have not yet been processed at all by time $\tau$. This is without loss of generality as it does not change the number of alive chunks at time $\ts$. We construct a solution of (\hyperlink{localchunklp}{$\chunklp(\ts)$}) by, for a chunk $c \in C$ with $\tau \ge r_{j(c)}$, defining
    $$
        x_c = \begin{cases}
            1 & \text{ if } c \in J_c^*\\
            0 & \text{ otherwise.}
        \end{cases}
    $$
    Clearly, the constructed solution has objective value $|J^*_c(\tau)|$. Hence, it only remains to argue that the solution is feasible. For the sake of contradiction, assume that the solution is infeasible and let $S$ be a set with a violated constraint. That is,
    $$
    \sum_{c \in S} \min(p_c, e(S)) \cdot x_c < e(S).
    $$
    Note that this immediately implies $\min\{p_c,e(S)\} = p_c$ for all $c \in J^*_c(\tau)$. Hence, plugging in the definition of excess,
    \begin{align*}
        \sum_{c \in S} \min(p_c, e(S)) \cdot x_c =  \sum_{c \in S} p_c \cdot x_c < e(S) = p(S) - (\tau-\ell_S).
    \end{align*}
    Rearranging and using the definition of the $x_c$'s yields
    \begin{equation}
        \label{eq:lp:relaxation}
        \tau-\ell_S < p(S) - \sum_{c \in S \cap J^*(\tau)} p_c.
    \end{equation}
    Now, note that the right-hand side is exactly the volume of chunks in $S\setminus J^*_c(\tau)$ that the optimal solution for the given instance of the \ofts~finishes by time $\tau$.  By definition of $\ell_S = \min_{c \in S} r_{j(c)}$, the chunks in $S\setminus J^*_c(\tau)$ cannot be processed before time $\ell_S$. Hence, the optimal schedule needs to process the complete volume of those chunks during $[\ell_S,\tau]$. However, this is a contradiction to~\eqref{eq:lp:relaxation}.
\end{proof}

\subsection{Justification of~\Cref{asm:reduced-property}: Reduced Instances}
\label{sec:reduced-instance}

Fix a time $\ts$ for which we want to show $|J(\ts)| \le \mathcal{O}(\ops_1 \cdot \ops_2) \cdot |J^*(\ts)|$ as in~\Cref{sec:local}. Our proofs in~\Cref{sec:excess-lemmas} rely on the~\Cref{asm:reduced-property}, which we restate here for convenience.

\asmReduced*

In this section, we argue that the assumption is w.l.o.g.~within the context of proving local competitiveness at time $\ts$ by exploiting the concept of \emph{reduced instances} as introduced in~\cite{GuptaKPW26}.
Assume we want to prove that $|J(\ts)| \le \mathcal{O}(\ops_1 \cdot \ops_2) \cdot |J^*(\ts)|$ for an instance $\mathcal{I}$ that does not satisfy~\Cref{asm:reduced-property}. We show that there exists a reduced instance $\reduced(\ts)$ with the same job and chunk sets as instance $\mathcal{I}$ that satisfies the following three properties:
\begin{enumerate}
    \item Instance $\reduced(\ts)$ satisfies~\Cref{asm:reduced-property}.
    \item $|J^*(\ts)| \le |J^*_{\mathrm{red}}(\ts)|$, where $J^*_{\mathrm{red}}(\ts)$ is the set of alive jobs at time $\ts$ in the optimal solution for instance $\reduced(\ts)$.
    \item $J(\ts) = \bA(\ts)$, where $\bA(\ts)$ is the set of active chunks at time $\ts$ in the algorithm's schedule for instance $\reduced(\ts)$.
\end{enumerate}

The existence of $\reduced(\ts)$ then implies that~\Cref{asm:reduced-property} is without loss of generality:
Since $\reduced(\ts)$ satisfies~\Cref{asm:reduced-property}, our proofs in~\Cref{sec:general} imply $|\bA(\ts)| \le \mathcal{O}(\ops_1 \cdot \ops_2) \cdot |J^*_{\mathrm{red}}(\ts)|$. Using the second and third property above, this implies 
$$
|J(\ts)| = |\bA(\ts)| \le \mathcal{O}(\ops_1 \cdot \ops_2) \cdot |J^*_{\mathrm{red}}(\ts)| \le  \mathcal{O}(\ops_1 \cdot \ops_2) \cdot |J^*(\ts)| \ .
$$
Hence, proving $|J(\ts)| \le \mathcal{O}(\ops_1 \cdot \ops_2) \cdot |J^*(\ts)|$ for instances that satisfy~\Cref{asm:reduced-property} implies the same for all instances. 

We continue by defining $\reduced(\ts)$ and proving that is satisfies the three properties above. The definition and proofs are chunk-based variants of the job-based equivalents given in~\cite{GuptaKPW26}.

The reduced instance $\reduced(\ts)$ uses the same jobs, operations and release dates as the original instance $\mathcal{I}$, and only adjusts the processing times of the operations. The following definition makes this more precise.

\begin{definition}[Reduced Instance $\reduced(\ts)$]
We define a reduced instance $\reduced(\ts)$ as follows.
For each operation $q$ that is part of a chunk $c \in A(\ts)$, we set 
\[
    p'_{q} := \min\{p_{q}, \max_{t \in [r_{c},\ts]} 2^{k(t) + 1} - \varepsilon \} \ ,
\]

where $r_c$ is the time when the chunk $c$ becomes active, which depends on the schedule constructed by the algorithm for the original instance~$\mathcal{I}$.
\end{definition}

\begin{claim}
 The reduced instance $\reduced(\ts)$ satisfies~\Cref{asm:reduced-property}.
\end{claim}

\begin{proof}
Fix an arbitrary $k$ and consider a chunk $c$ that becomes active at time $r_c > t_{\ge k}$. By definition of $\reduced(\ts)$, all operations $q \in c$ have a reduced processing time $p_q' \le 2^{k'+1} - \varepsilon < 2^{k'+1}$ for $k' = \max_{t \in r_c,\ts} k(t)$. Since $p_q' < 2^{k'+1}$ for all $q \in c$, all such operations are of a class $\le k'$. By~\Cref{def:chunk}, this implies that chunk $c$ is of class $\le k'$ as well. Finally, the fact that $r_c > t_{\ge k}$ and the definition of $t_{\ge k}$ (cf.~\Cref{def:geqk}) imply that $k' < k$. Hence, $c$ is of class $< k$.
\end{proof}

Next, we observe that $\reduced(\ts)$ satisfies the second property, i.e., $|J^*(\ts)| \le |J^*_{\mathrm{red}}(\ts)|$. This is easy to see as reducing processing times of operations can only help an optimal solution to reduce its number of alive jobs at time $\ts$.

Finally, it remains to argue that $\reduced(\ts)$ satisfies the third property, i.e.,  $J(\ts) = \bA(\ts)$. To this end, we will show that the algorithm behaves the same on the reduced instance until time $\ts$ as on the original instance.

First, note that chunks $c \in \apart(\ts) = J(\ts) \setminus \afull(\ts)$, i.e., partial chunks at time $\ts$, will never be affected by the reduction of processing times: if a chunk $c$ is partial at time $\ts$, then it is processed for at least one time unit during $[r_c,\ts]$. Hence, $\max_{t \in [r_c,\ts]} k(t) \ge k_c$. Assuming sufficiently small $\varepsilon$, this implies $p_q = p'_q$ for all operations $q$ that belong to chunk $c$. Hence, we also get $k_c = k_c'$ for all chunks $c$ that are partial at $\ts$ in the original instance, where $k_c'$ denotes the class of $c$ in the reduced instance.

Next, note that if the processing time of operation $q$ in a chunk $c$ is reduced, then the processing time 
of the first operation $q_1$ in the chunk $c$ also has to be reduced as $q_1$ has maximum class in $c$. After the reduction, $q_1$ still has maximum class in $c$.
By reducing the processing time of the first operation 
$q_1$ in a chunk $c$, we also affect the class of $c$
which is determined by the first operation in $c$. In particular, the class of $c \in J(\ts)$ in the reduced instance denoted by $k'_c$ is at most $\max_{t \in [r_c,\ts]} k(t)$. 

With the next lemma, we formally prove that the algorithm behaves the same on the reduced instance until time $\ts$ as on the original instance.

\begin{lemma}
    \label{lem:reduced:equiv}
    Let $\bapart(t)$, $\bA(t)$, $\bAfull(t)$ denote the sets $\apart(t)$, $J(t)$, and $\afull(t)$ of chunks in the reduced instance. Furthermore, let $p'_c(t)$ denote the remaining processing time of chunk $c$ at time $t$ in the algorithm's schedule in the reduced instance.
    At any time $t \le \ts$, we have $\apart(t) = \bapart(t)$, $J(t) = \bA(t)$ and $\afull(t) = \bAfull(t)$, as well as $p_c(t) = p'_c(t)$ for all $c \in \apart(t) = \bapart(t)$. In particular, this implies $|J(\ts)| = |\bA(\ts)|$.
\end{lemma}

\begin{proof}
    We show the statement via induction over $0 \le t \le \ts$. For $t = 0$, the statement holds trivially. 

    Assume the statement holds for time $t-1$. We show that it then also holds for time $t$ by distinguishing between the following two cases:
    \begin{enumerate}
        \item The algorithm does not move a chunk to $\apart(t-1)$ at time $t-1$
        in the original instance.
        \item The algorithm moves a chunk to $\apart(t-1)$ at time $t-1$ in the original instance. 
    \end{enumerate}

    \textbf{Case 1:} 
    We first argue that in this case the algorithm also does not move a chunk to $\bapart(t-1)$ at time $t-1$ in the reduced instance.

    The fact that the algorithm does not move a chunk in the original instance implies $|\afull(t-1)|<\frac{|J(t-1)|}{4}$ or that $\min_{c \in \afull(t-1)} k_c \ge \min_{c \in \apart(t-1)} k_c$. In the former case, the induction hypothesis implies $|\bAfull(t-1)|<\frac{|\bA(t-1)|}{4}$ which means that the algorithm also does not move a chunk in the reduced instance.  

    In the latter case, we can first observe that $\min_{c \in \apart(t-1)} k_c = \min_{c \in \bapart(t-1)} k'_c$ since, as mentioned before, the processing times of operations in partial chunks are never reduced. 
    Then, using notation  $k'_{\min} = \min_{c \in \bapart(t-1)} k'_c$, we observe that for any $c\in A(t-1)$, $\max_{t' \in [r_c,\ts]} 2^{k(t')+1}-\varepsilon \ge \max_{t' \in [t-1,\ts]} 2^{k(t')+1}-\varepsilon \ge 2^{k'_{\min}+1} -\varepsilon$, where the last inequality follows $\min_{c \in \apart(t-1)} k_c = \min_{c \in \bapart(t-1)} k'_c$  and from the fact that the algorithm processes some chunk in $\apart(t-1)$ during $[t-1,t]$ in the original instance. In particular this implies $k_c' \ge k'_{\min}$ for all $c \in \afull(t-1)$.
    Since we have $\afull(t-1) = \bAfull(t-1)$ by induction, this also implies $k'_c \ge k'_{\min}$ for all $c \in \bAfull(t-1)$. Hence, the algorithm will not move a chunk to $\bapart(t-1)$ at time $t-1$ in the reduced instance.

    Using that the algorithm does not move a chunk at time $t-1$ in either instance and that the classes of partial chunks are never reduced, we can conclude that the algorithm works on the same chunk during $[t-1,t]$. Hence, $\apart(t) = \bapart(t)$ follows from $\apart(t-1) = \bapart(t-1)$. Using that $p_c(t-1) = p'_c(t-1)$ for all $c \in \apart(t-1) = \bapart(t-1)$ holds by induction hypothesis, we also get  $p_c(t) = p'_c(t)$ for all $c \in \apart(t) = \bapart(t)$. Finally, $\bA(t-1) = J(t-1)$ together with $\apart(t) = \bapart(t)$ also implies $J(t) = \bA(t)$ and $\afull(t) = \bAfull(t)$.

    \textbf{Case 2:} Assume that the algorithm moves some chunk $c$ from to $\apart(t-1)$ at time $t-1$ in the original instance. We argue that the same chunk is moved to $\bapart(t-1)$ at $t-1$ in the reduced instance. Since $c$ is moved in the original instance, we have $|\afull(t-1)|\ge\frac{|J(t-1)|}{4}$ which implies $|\bAfull(t-1)|\ge\frac{|\bA(t-1)|}{4}$ by the induction hypothesis. 

    Next, we can observe that $c$ being moved in the original instance implies $k_c < \min_{c' \in \apart(t-1)} k_{c'}$. Since chunks that are partial at some point during $[0,\ts]$ are never affected by the processing time reduction, this implies  $k'_c < \min_{c' \in \apart(t-1)} k'_{c'} = \min_{c' \in \bapart(t-1)} k'_{c'}$, where the equality uses that $\apart(t-1) = \bapart(t-1)$ by induction hypothesis. Note that this already proves that the algorithm moves \emph{some} chunk to $\bapart(t-1)$ in the reduced instance, but it remains to show that this chunk is $c$. 
    Consider some chunk $c' \in \bAfull(t-1)\setminus\{c\}$. If $p_{c'} = p'_{c'}$, then the chunk $c'$ was not affected by the processing time reduction. Since the algorithm, in both instances, uses the same order over $\afull(t-1)=\bAfull(t-1)$ to decide which chunk to move, it has to select $c$ over $c'$ in both instances\footnote{This assumes that the algorithm uses the same tie-breaking rule for chunks that have the same class and the same first operation size on both instances.}. 
    
    If $c'$ was affected by the processing time reduction, then the fact that the algorithm processes $c$ during $[t-1,t]$ in the original instance implies that 
    \begin{enumerate}
        \item the class $k'_{c'}$ of $c'$ in the reduced instance is at least $k_c = k'_c$ and
        \item By the reduction rule $p_{q'} \ge 2^{k'_c+1}-\varepsilon > p_q$
        where $q'$ and $q$ are the first operations of the chunks $c'$ and $c$, respectively. Note that the latter inequality holds for a sufficiently small $\varepsilon$.
    \end{enumerate}
    These two facts imply that the algorithm will also move $c$ over $c'$ in the reduced instance: The first fact implies that $c'$ cannot have a smaller class than $c$ in the reduced instance, and the second fact implies that $c'$ loses the tiebreaker against $c$ in case both chunks have the same class in the reduced instance.
    
    Using that the algorithm moves the same chunk at time $t-1$ in both instances, we can conclude that the algorithm works on the same chunk during $[t-1,t]$ in both instances. Hence, $\apart(t) = \bapart(t)$ follows from $\apart(t-1) = \bapart(t-1)$. Using that $p_c(t-1) = p'_c(t-1)$ for all $c \in \apart(t-1) = \bapart(t-1)$ holds by induction hypothesis, we also get  $p_c(t) = p'_c(t)$ for all $c \in \apart(t) = \bapart(t)$. Finally, $\bA(t-1) = J(t-1)$ together with $\apart(t) = \bapart(t)$ also implies  $J(t) = \bA(t)$ and $\afull(t) = \bAfull(t)$.
\end{proof}

\end{document}